\documentclass[journal,tbtags]{IEEEtran}

\usepackage{mathtools,amsfonts,amssymb,dsfont}

\usepackage{mathtools} 

\usepackage[sort,compress]{cite}
\usepackage{url}
\usepackage{amsfonts,dsfont,amssymb,bm}

\usepackage[amsmath,standard,thmmarks]{ntheorem}

\usepackage{subcaption}
\usepackage{microtype}

\usepackage{algorithm}
\usepackage[noend]{algpseudocode}
\usepackage{tikz}
\tikzset{
  agent/.style={draw, fill=blue!25!white,circle, minimum size=2mm,
  font=\footnotesize},
  x = 1.5cm, y=1.5cm,
  every loop/.style={},
}

\usepackage{comment}
\excludecomment{delete}

\newcommand*\reals{\mathds{R}}
\newcommand\Sp[1]{\reals^{d_{#1}}}
\newcommand*\MAT{}
\newcommand*\PR{\mathds{P}}

\newcommand*\brMAT[1]{\breve{\MAT{#1}}}

\DeclareMathOperator\DARE{DARE}
\DeclareMathOperator\VAR{var}
\DeclareMathOperator\TR{Tr}

\let\citep\cite

\newcommand*\DELETE[1]{}

\newcommand*\EXP{\mathbb{E}}

\newcommand*\TRANS{{\mathpalette\doTRANS\empty}}
\makeatletter
\newcommand*\doTRANS[2]{\raisebox{\depth}{$\m@th#1\intercal$}}
\makeatother

\newcommand*\VVEC{}

\DeclareMathOperator{\VEC}{vec}
\DeclareMathOperator{\DIAG}{diag}

\newcommand*\iell{i^\ell_*}
\newcommand*\Iell{\ell,\iell}
\newcommand*\tildeO{\tilde{\mathcal{O}}}

\usepackage{xcolor}

\newcommand{\red}[1]{\textcolor{black}{#1}}
\newcommand{\blue}[1]{\textcolor{black}{#1}}

\usepackage [english]{babel}
\usepackage [autostyle, english = american]{csquotes}
\MakeOuterQuote{"}

\title{Scalable regret for learning to control network-coupled
subsystems with unknown dynamics}
\author{Sagar Sudhakara, Aditya Mahajan, Ashutosh Nayyar, and Yi Ouyang%
  \thanks{Sagar Sudhakara and Ashutosh Nayyar are with the Department of Electrical and Computer Engineering, University of Southern California, Los Angeles, CA, USA. (email: sagarsud@usc.edu, ashutosn@usc.edu)}%
  \thanks{Aditya Mahajan is with the department of Electrical and Computer Engineering, McGill University, Montreal, QC, Canada. (email: aditya.mahajan@mcgill.ca)}%
  \thanks{Yi Ouyang is with Preferred Networks America, Burlingame, CA, USA (email: ouyangyi@preferred-america.com)}%
\thanks{The work of Aditya Mahajan was supported in part by  the Innovation for Defence Excellence and Security (IDEaS) Program of the Canadian Department of National Defence through grant CFPMN2-30.}%
}

\begin{document}
\maketitle

\begin{abstract}
 We consider the problem of controlling  an unknown linear quadratic Gaussian (LQG) system consisting of multiple subsystems connected over a network. Our goal is to minimize and quantify the regret (i.e. loss in performance) of our  strategy with respect to an oracle who knows the system model. Viewing the interconnected subsystems globally and directly using existing LQG learning algorithms for the global system results in a regret that increases super-linearly with the number of subsystems. Instead,  we propose a new Thompson sampling based learning algorithm which exploits the structure of the underlying network. We show that the expected regret of the proposed algorithm is bounded by  $\tilde{\mathcal{O}} \big( n \sqrt{T} \big)$ where $n$ is the number of subsystems,  $T$ is the time horizon and the $\tilde{\mathcal{O}}(\cdot)$ notation hides logarithmic terms in $n$ and~$T$. Thus, the regret scales linearly with the number of subsystems. We present numerical experiments to illustrate the salient features of the proposed algorithm.
\end{abstract}


\section{Introduction}

Large-scale systems comprising of multiple subsystems connected over a network arise in a number of applications including power systems, traffic networks, communication networks and some economic systems \cite{sandell1978}. A common feature of such systems is the coupling in their subsystems' dynamics and costs, i.e., the state evolution and local costs of one subsystem depend not only on its own state and control action but also on the states and control actions of other subsystems in the network.  Analyzing various aspects of the behavior of such systems and designing control strategies for them under a variety of settings have been long-standing problems of interest in the systems and control literature~\citep{lunze1986dynamics,sundareshan1991qualitative,yang1995structural,hamilton2012patterned,arabneydi2015mft,arabneydi2016mft}. 
However, there are still many unsolved challenges, especially on the interface between learning and control in the context of these large-scale systems.

In this paper, we investigate the problem of designing control strategies for large-scale network-coupled subsystems when some parameters of the system model are not known. {Due to the unknown parameters, the control problem is also a learning problem. We adopt a reinforcement learning framework} for this problem with the goal of  minimizing and quantifying the regret (i.e. loss in performance) of our learning-and-control strategy with respect to the optimal control strategy based on the complete knowledge of the system model.

The networked system we consider follows linear dynamics with quadratic costs
and Gaussian noise. Such linear-quadratic-Gaussian (LQG) systems are one of
the most commonly used modeling framework in numerous control applications.
Part of the appeal of LQG models is the simple structure  of the optimal
control strategy when the system model is completely known---the optimal
control action  in this case  is a linear or affine function of the state---which makes the optimal strategy easy to identify and easy to implement. If some
parameters of the model are not fully known during the design phase or may
change during operation, then it is better to design a strategy that learns
and adapts online. Historically, both adaptive
control~\citep{astrom1994adaptive} and
reinforcement learning~\citep{bradtke1993reinforcement, bradtke1994adaptive} have been used to design asymptotically
optimal learning algorithms for such LQG systems. In recent years, there has
been considerable interest in analyzing the transient behavior of such
algorithms which can be quantified in terms of the \emph{regret} of the
algorithm  as a function of time. This allows one to assess, as a function
of time, the performance of a  learning algorithm   compared to an oracle who
knows the system parameters upfront.

Several learning algorithms have been proposed for LQG systems \citep{campi1998adaptive,
abbasi2011regret, faradonbeh2017finite,cohen2019learning,abeille2020efficient, dean2018regret,mania2019certainty,faradonbeh2020input,simchowitz2020naive,faradonbeh2020adaptive,ouyang2017control,ouyang2019posterior,abeille2018improved}, and in most cases the regret is shown to be bounded by  $\tilde
{\mathcal{O}}( d_x^{0.5} (d_x + d_u) \sqrt{T})$, where $d_x$ is the dimension
of the state, $d_u$ is the dimension of the controls, $T$ is the time
horizon, and the $\tilde{\mathcal{O}}(\cdot)$ notation hides logarithmic
terms in~$T$.
Given the lower bound of $\tilde{\Omega}(d_x^{0.5} d_u
\sqrt{T})$ (where $\tilde \Omega(\cdot)$ notation hides  logarithmic terms in~$T$) for regret in LQG systems identified in a recent work \cite{simchowitz2020naive}, the regrets of the existing algorithms have near optimal scaling in terms of time and dimension. However, when directly applied to a networked system with  $n$ subsystems, these algorithms would incur $\tilde{\mathcal{O}}( n^{1.5} d_x^{0.5} ({d_x + d_u}) \sqrt{T})$ regret because the effective
dimension of the state and the controls is $nd_x$ and $nd_u$, where  $d_x$ and $d_u$ are the dimensions of each subsystem.
This super-linear dependence on $n$ is prohibitive in large-scale networked systems  because the regret per subsystem (which is $\tildeO(\sqrt{n})$) grows with the number of subsystems.

The learning algorithms mentioned above are for a general LQG system and do not take into account any knowledge of the underlying network structure.  
Our main contribution is to show that by exploiting the structure of the network model, it is possible to design learning algorithms for large-scale network-coupled subsystems where the regret does not grow super-linearly in the number of subsystems.
In particular, we utilize a spectral decomposition technique, recently proposed in~\cite{gao2019networked}, to decompose the large-scale system into $L$ decoupled systems, where $L$ is the rank of the coupling matrix corresponding to the underlying network. 
Using the decoupled systems, we propose a Thompson sampling based algorithm with $\tilde{\mathcal{O}}( n\allowbreak d_x^{0.5} ({d_x + d_u}) \sqrt{T})$ regret bound.

\paragraph{Related work}

Broadly speaking, three classes of low-regret learning algorithms have been
proposed for LQG systems:
certainty equivalence (CE) based algorithms, optimism in the face of uncertainty (OFU) based algorithms, and Thompson sampling (TS) based
algorithms.
CE is a classical adaptive control algorithm~\citep{astrom1994adaptive}.  
Recent papers~\citep{dean2018regret,mania2019certainty,faradonbeh2020input,simchowitz2020naive,faradonbeh2020adaptive} have established near optimal high probability bounds on regret for CE-based algorithms.
OFU-based algorithms are inspired by the OFU principle for multi-armed
bandits~\citep{auer2002finite}. 
Starting with the work of~\citep{campi1998adaptive,
abbasi2011regret}, most of the papers following the OFU
approach~\citep{faradonbeh2017finite,cohen2019learning,abeille2020efficient} also provide similar high probability regret bounds.
TS-based algorithms are inspired by TS 
algorithm for multi-armed bandits~\citep{agrawal2012analysis}. 
Most papers following this
approach~\citep{ouyang2017control,ouyang2019posterior,abeille2018improved,faradonbeh2020adaptive} establish bounds on expected Bayesian regret of similar near-optimal orders.
As argued earlier, most of these papers show that the regret scales super-linearly with the number of subsystems and are, therefore, of limited value for large-scale systems.

  There is an emerging literature on learning algorithms for networked systems both for 
  LQG models~\cite{wang2020distributed,jing2021learning,li2019distributed}
  and
  MDP models~\cite{zhang2018fully,zhang2019decentralized,NIPS2014_0deb1c54,chen2021efficient}.
The papers on LQG models propose distributed value- or policy-based learning algorithms and analyze their convergence properties, but they do not characterize their regret. Some of the papers on MDP models~\cite{NIPS2014_0deb1c54,chen2021efficient} do characterize regret bounds for OFU and TS-based learning algorithms but these bounds are not directly applicable to the LQG model considered in this paper.

An important special class of network-coupled systems is mean-field coupled subsystems~\cite{HuangCainesMalhame_2007,LasryLions_2007}.
There has been
considerable interest in reinforcement learning for mean-field
models~\citep{subramanian2019reinforcement, subramanian2020multi, uz2020reinforcement}, but most of the literature does not consider regret. The basic mean-field coupled model can be viewed as a special case of the network-coupled subsystems considered in this paper (see Sec.~\ref{sec:mf}). In a preliminary version of this paper~\citep{gagrani2020thompson}, we proposed a TS-based algorithm for mean-field coupled subsystems which has a $\tildeO( (1+1/n)\sqrt{T})$ regret per subsystem.
The current paper extends the TS-based algorithm to general network-coupled subsystems and establishes scalable regret bounds for arbitrarily coupled networks.

\paragraph{Organization}
The rest of the paper is organized as follows. In Section~\ref{sec:model},
we introduce the model of network-coupled subsystems.
In Section~\ref{sec:spectral}, we summarize the spectral decomposition idea and the resulting scalable method for synthesizing optimal control strategy when the model parameters are known.
Then, in Section~\ref{sec:learning}, we consider the learning problem for unknown network-coupled subsystems and present a TS-based learning algorithm with scalable regret bound. We subsequently provide regret analysis in Section~\ref{sec:proof} and numerical experiments in Section~\ref{sec:examples}. We conclude in Section~\ref{sec:conclusion}.

\paragraph{Notation}
The notation $\MAT A = [a^{ij} ]$ means that $a^{ij}$ is
the $(i, j)$th element of the matrix $\MAT A$. For a matrix $\MAT A$, $\MAT A^\TRANS$ denotes its transpose. Given matrices (or vectors) $A_1$, \dots, $A_n$ with the same number of rows, $[A_1, \dots, A_n]$ denotes the matrix formed by horizontal concatenation.
For a random vector $v$, $\VAR(v)$ denotes its covariance matrix. The notation $\mathcal{N}(\mu, \Sigma)$ denotes the  multivariate Gaussian distribution with mean vector $\mu$ and covariance matrix $\Sigma$.

For stabilizable $(\MAT A,\MAT B)$ and positive definite matrices $Q, R$, DARE$(A,B,Q,R)$ denotes the unique positive semidefinite solution of the discrete time algebraic Riccati equation (DARE), which is given as 
\[S=A^\TRANS SA - (A^\TRANS SB)(R+B^\TRANS SB)^{-1}(B^\TRANS SA)+Q.\]

\section{Model of network-coupled subsystems}\label{sec:model}
We start by describing a minor variation of a model of network-coupled
subsystems proposed in~\cite{gao2019networked}. The model in ~\cite{gao2019networked} was described in continuous time. We translate the model and the results to discrete time.

\subsection{System model}

\subsubsection{Graph stucture}
Consider a network consisting of $n$ subsystems/agents connected over an undirected weighted graph denoted by $\mathcal{G}(N,E,W_{\mathcal{G}})$, where $N=\{1,\dots,n\}$ is the set of nodes, $E\subseteq N\times N$ is the  set of edges, and $W_{\mathcal{G}} =[w^{ij}]\in \mathds{R}^{n\times n}$ is the weighted adjacency matrix. Let $M=[m^{ij}] \in \mathds{R}^{n\times n}$ be  a symmetric coupling matrix corresponding to the underlying graph $\mathcal{G}$. For instance, $M$ may represent the underlying adjacency matrix (i.e., $M=W_{\mathcal{G}}$) or the underlying Laplacian matrix (i.e., $M=\DIAG(W_{\mathcal{G}}\mathds{1}_n)-W_{\mathcal{G}}$). 

\subsubsection{State and dynamics}
The states and control actions of agents take values in $\Sp{x}$ and $\Sp{u}$, respectively. For agent $i \in N$, we use $x^i_t \in \Sp{x}$ and $u^i_t \in \Sp{u}$ to denote its state and
control action at time~$t$. 

The system starts at a random initial state $x_1 = (x^i_1)_{i \in N}$, whose
components are independent across agents. For agent~$i$,
the initial state $x^i_1 \sim \mathcal N(0, \MAT \Xi^i_1)$, and at any time~$t \ge
1$, the state evolves according to 
\begin{equation}\label{eq:dynamics}
  x^i_{t+1} = \MAT A x^i_{t} + \MAT B u^i_{t} + 
  \MAT D  x^{\mathcal{G},i}_t + \MAT E u^{\mathcal{G},i}_t + 
  w^i_t,
\end{equation}
where $x^{\mathcal{G},i}_t$ and $u^{\mathcal{G},i}_t$ are the locally
perceived influence of the network  on the state of agent~$i$ and are given by
\begin{equation}\label{local_perceived}
    x^{\mathcal{G},i}_t=\sum_{j \in N}m^{ij}x^j_{t}
    \quad\text{and}\quad
    u^{\mathcal{G},i}_t=\sum_{j \in N}m^{ij}u^j_{t}
    ,
\end{equation}
$\MAT A$, $\MAT B$, $\MAT D$, $\MAT E$
are matrices of appropriate dimensions, and $\{w^i_t\}_{t \ge 1},i \in N,$ are
i.i.d.\@ zero-mean Gaussian processes which are independent of each other and
the initial state. In particular, $w^i_t \in \Sp{x}$ and $w^i_t \sim \mathcal
N(0, \MAT W)$.  We call $x^{\mathcal{G},i}_t$ and $u^{\mathcal{G},i}_t$ the
\emph{network-field} of the states and control actions at node~$i$ at time~$t$.

Thus, the next state of agent~$i$ depends on its
current local state and control action, the current network-field of the states and
control actions of the system, and the current local noise.

We follow the same atypical representation of the ``vectorized'' dynamics as used
in~\cite{gao2019networked}. Define $x_t$ and $u_t$ as the global state and
control actions of the system:
\[
  x_t=[x^1_t,\dots.,x^n_t]
    \quad\text{and}\quad
    u_t=[u^1_t,\dots.,u^n_t].
\]
We also define $w_t=[w^1_t,\dots.,w^n_t]$.
Similarly, define $x^{\mathcal{G}}_t$ and $u^{\mathcal{G}}_t$ as the global
network field of states and actions:
\[
  x^{\mathcal{G}}_t=[x^{\mathcal{G},1}_t,\dots.,x^{\mathcal{G},n}_t]
  \quad\text{and}\quad
  u^{\mathcal{G}}_t=[u^{\mathcal{G},1}_t,\dots.,u^{\mathcal{G},n}_t].
\]
Note that $x_t, x^{\mathcal{G}}_t,w_t\in \mathbb{R} ^{d_x\times n}$ and $u_t,
u^{\mathcal{G}}_t\in \mathbb{R} ^{d_u\times  n}$ are matrices and not vectors. 
The global system dynamics may be written as:
\begin{equation}\label{eq:dynamics2}
     x_{t+1} = \MAT A x_{t} + \MAT B u_t + 
  \MAT D  x^\mathcal{G}_t + \MAT E u^\mathcal{G}_t + 
  w_t. 
\end{equation}
Furthermore, we may write 
\[
  x^{\mathcal{G}}_t = x_tM^\TRANS=x_tM
  \quad\text{and}\quad
  u^{\mathcal{G}}_t=u_tM^\TRANS=u_tM.
\]
\subsubsection{Per-step cost}
At any time $t$ the system incurs a per-step cost given by
\begin{equation}\label{eq:cost}
  c(x_t,u_t)=\sum_{i\in N}\sum_{j\in N}[h_x^{ij}(x^i_t)^\TRANS Q(x^j_t)+h_u^{ij}(u^i_t)^\TRANS R(u^j_t)]
\end{equation}
where $Q$ and $R$ are matrices of appropriate dimensions and $h_x^{ij}$ and
$h_u^{ij}$ are real valued weights. Let $H_x = [ h_x^{ij} ]$ and $H_u= [h_u^{ij}]$. It is
assumed that the weight matrices $H_x$ and $H_u$ are polynomials of $M$, i.e.,
\begin{equation}\label{eq:poly}
  H_x=\sum_{k=0}^{K_x}q_kM^k
  \quad\text{and}\quad
  H_u=\sum_{k=0}^{K_u}r_kM^k
\end{equation}
where $K_x$ and $K_u$ denote the degrees of the polynomials and
$\{q_k\}^{K_x}_{k=0}$ and $\{r_k\}_{k=0}^{K_u}$ are real-valued coefficients. 

The assumption that $H_x$ and $H_u$ are polynomials of $M$ captures the intuition that the
per-step cost respects the graph structure.  In the special case when $H_x =H_u=I$,
the per-step cost is decoupled across agents. When $H_x =H_u = I + M$,
 the per-step cost captures a cross-coupling between one-hop neighbors.
Similarly, when $H_u = I + M + M^2$, the per-step cost captures a cross-coupling
between one- and two-hop neighbors. See~\cite{gao2019networked} for more
examples of special cases of the per-step cost defined above.

\subsection{Assumptions on the model}

Since $M$ is real and symmetric, it has real eigenvalues. 
Let $L$ denote the rank of $M$
and $\lambda^1,\dots.,\lambda^L$ denote the non-zero eigenvalues. For ease of notation, for
$\ell \in {1,\dots,L}$, define
\[
  q^\ell=\sum_{k=0}^{K_x}q_k(\lambda^\ell)^k 
  \quad\text{and}\quad
  r^\ell=\sum_{k=0}^{K_u}r_k(\lambda^\ell)^k ,
\]
where $\{q_k\}^{K_x}_{k=0}$ and $\{r_k\}^{K_u}_{k=0}$ are the coefficients in~\eqref{eq:poly}.
Furthermore, for $\ell \in \{1, \dots, L \}$, define:
\[
  \MAT A^\ell = {\MAT A} + \lambda^\ell{\MAT D}
  \quad\text{and}\quad
  \MAT B^\ell = {\MAT B} + \lambda^\ell{\MAT E}.
\]

We impose the following assumptions:
\begin{description}
  \item[(A1)] The systems $(\MAT A, \MAT B)$ and $\{ (\MAT A^\ell, \MAT
    B^\ell) \}_{\ell=1}^L$ are stabilizable.
    \item[(A2)] The matrices $Q$ and $R$ are symmetric and positive definite.
    \item[(A3)] The parameters $q_0$, $r_0$, $\{q^\ell\}_{\ell=1}^L$, and $\{r^\ell\}_{\ell = 1}^L$ are strictly
    positive.
\end{description}

Assumption (A1) is a standard assumption. Assumptions (A2) and (A3) ensure
that the per-step cost is strictly positive.

\subsection{Admissible policies and performance criterion}
 There is a system operator who has access to the state and action histories of all
agents  and who selects the agents' control actions according to a
deterministic or randomized (and potentially history-dependent) policy
\begin{equation}
   u_t = \pi_t( x_{1:t},  u_{1:t-1}).
\end{equation}

Let $\VVEC \theta^\TRANS = [\MAT A, \MAT B, \MAT D, \MAT E]$ denote the
parameters of the system dynamics. The performance of any policy $\pi =
(\pi_1, \pi_2, \dots)$ is measured by the long-term average cost given by
\begin{equation} \label{eq:cost_performance}
  J(\pi; \VVEC \theta) = \limsup_{T \to \infty} \frac {1}{T}
    \EXP^\pi\biggl[ \sum_{t=1}^{T} c(x_t, u_t)  \biggr].
\end{equation}
Let $J(\VVEC \theta)$ denote the minimum of $J(\pi; \VVEC \theta)$ over all
policies.

We are interested in the setup where the graph coupling matix $M$, the cost coupling
matrices $H_x$ and $H_u$, and the cost matrices $Q$ and $R$ are known but the system dynamics $\VVEC \theta$ are
unknown and there is a prior distribution on $\VVEC \theta$. 
The Bayesian \emph{regret} of a policy $\pi$ operating for a
horizon $T$ is defined as 
\begin{equation}
  R(T; \pi) \coloneqq
  \EXP^\pi \biggl[ \sum_{t=1}^T c(x_t, u_t) 
  - T J(\VVEC \theta) \biggr],
\end{equation}
where the expectation is with respect to the prior on $\theta$, the noise
processes, the initial
conditions, and the potential randomizations done by the policy~$\pi$.

\section{Background on spectral decomposition of the system}\label{sec:spectral}
In this section, we summarize the main results of~\cite{gao2019networked},
translated to the discrete-time model used in this paper.

The spectral decomposition described in~\cite{gao2019networked} relies on the spectral
factorization of the graph coupling matrix~$M$.
Since $M$ is a real and symmetric matrix with rank $L$, we can write it as
\begin{equation}\label{eq:spectral_fact}
  M=\sum_{\ell=1}^{L}\lambda^\ell v^\ell (v^{\ell})^\TRANS,
\end{equation}
where $(\lambda^1,\dots,\lambda^L)$ are the non-zero eigenvalues of $M$ and
$(v^1,\dots,v^L)$ are the corresponding orthonormal eigenvectors.

We now present the decomposition of the dynamics and the cost based
on~\eqref{eq:spectral_fact} as described in~\cite{gao2019networked}.

\subsection{Spectral decomposition of the dynamics and per-step cost}
For $\ell \in \{1,2,\dots,L\}$, define \emph{eigenstates} and
\emph{eigencontrols} as
\begin{equation}\label{eq:xell}
    x^\ell_t=x_t v^\ell (v^{\ell})^\TRANS
    \quad\text{and}\quad
    u^\ell_t=u_t v^\ell (v^{\ell})^\TRANS,
\end{equation}
respectively. Furthermore, 
define \emph{auxiliary state} and \emph{auxiliary control} as
\begin{equation}\label{eq:xbreve}
    \breve x_t=x_t-\sum_{\ell=1}^{L} x^\ell_t
    \quad\text{and}\quad
    \breve u_t=u_t-\sum_{\ell=1}^{L} u^\ell_t,
\end{equation}
respectively. Similarly, define $w^\ell_t = w_t v^\ell (v^{\ell})^\TRANS$ and
$\breve w_t = w_t - \sum_{\ell = 1}^L w^\ell_t$. 
Let $x^{\ell,i}_t$ and $u^{\ell,i}_t$ denote the $i$-th column of $x^\ell_t$ and $u^\ell_t$ respectively; thus we can write
\[
  x^\ell_t=[x^{\ell,1}_t,\dots.,x^{\ell,n}_t]
  \quad\text{and}\quad
  u^\ell_t=[u^{\ell,1}_t,\dots.,u^{\ell,n}_t].
\]
Similar interpretations hold for $w^{\ell,i}_t$ and $\breve w^i_t$.
Following~\cite[Lemma~2]{gao2019networked}, we can show that for any $i \in N$,
\begin{equation}\label{eq:var}
  \VAR(w^{\ell,i}_t) = (v^{\ell,i})^2 W
  \quad\text{and}\quad
  \VAR(\breve w^i_t) = (\breve v^i)^2 W,
\end{equation}
where $(\breve v^i)^2 = 1 - \sum_{\ell=1}^L (v^{\ell,i})^2$. These covariances
do not depend on time because the noise processes are i.i.d.

Using the same argument as in~\cite[Propositions~1 and~2]{gao2019networked}, we can show the following.
\begin{proposition}\label{prop:spectral}
  For each node $i \in N$, the state and control action may be decomposed as
  \begin{equation}\label{eq:dec1}
    x^i_t=\breve x^i_t+\sum_{\ell=1}^{L} x^{\ell,i}_t 
    \quad\text{and}\quad
    u^i_t=\breve u^i_t+\sum_{\ell=1}^{L} u^{\ell,i}_t
  \end{equation}
  where the dynamics of eigenstate $x^{\ell,i}_t$   depend only on $u^{\ell,i}_t$ and $w^{\ell,i}_t$, and are given by
  \begin{equation}
    \label{eq:field-dynamics}
    x^{\ell,i}_{t+1} = \MAT (A+\lambda^\ell D) x^{\ell,i}_t + \MAT (B+\lambda^\ell E) u^{\ell,i}_t + w^{\ell,i}_t ,
  \end{equation}
  and the dynamics of the auxiliary state $\breve x^i_t$ depend  only on   $\breve u^i_t$ and $\breve w^i_t$  and are given by 
  \begin{equation}
    \label{eq:aux-dynamics}
    \breve x^i_{t+1}= \MAT A \breve x^i_t + \MAT B \breve u^i_t + \breve w^i_t.
  \end{equation}
    Furthermore, the per-step cost decomposes as follows:
  \begin{equation}
    c(x_t, u_t) = \sum_{i \in N} \biggl[ q_0 \breve c(\breve x^i_t, \breve u^i_t) 
    + \sum_{\ell=1}^L q^\ell c^\ell(x^{\ell,i}_t, u^{\ell,i}_t) \biggr]
  \end{equation}
  where\footnote{Recall that (A3) ensures that $q_0$ and
  $\{q^\ell\}_{\ell=1}^L$ are strictly positive.}
  \begin{align*}
    \breve c(\breve x^i_t,\breve u^i_t) &= 
     ({\breve x^{i}_t})^\TRANS \MAT Q \breve x^i_t 
     +  \frac{r_0}{q_0} ({\breve u^{i}_t})^\TRANS \MAT R \breve u^i_t, \\
    c^\ell(x^{\ell,i}_t, u^{\ell,i}_t) &=
    (x^{\ell,i}_t)^\TRANS \MAT Q x^{\ell,i}_t + \frac{r^\ell}{q^\ell}  (u^{\ell,i}_t)^\TRANS \MAT R u^{\ell,i}_t.
  \end{align*}
\end{proposition}

\subsection{Planning solution for network-coupled subsystems} \label{sec:planning}

We now present the main result of~\cite{gao2019networked}, which provides a
scalable method to synthesize the optimal control policy when the system
dynamics are known. 

Based on Proposition~\ref{prop:spectral}, we can view the overall system as
the collection of the following subsystems:
\begin{itemize}
  \item \emph{Eigen-system $(\ell,i)$}, $\ell \in \{1, \dots, L\}$ and $i \in
    N$ with state $x^{\ell,i}_t$, controls $u^{\ell,i}_t$,
    dynamics~\eqref{eq:field-dynamics}, and per-step cost $q^\ell
    c^\ell(x^{\ell,i}, u^{\ell,i})$.
  \item \emph{Auxiliary system~$i$}, $i \in N$, with state $\breve x^i_t$,
    controls $\breve u^i_t$, dynamics~\eqref{eq:aux-dynamics}, and per-step
    cost $\breve q_0 \breve c(\breve x^i_t, \breve u^i_t)$. 
\end{itemize}

Let $(\theta^\ell)^\TRANS = [ A^\ell, B^\ell ] \coloneqq [ (A + \lambda^\ell
D), (B + \lambda^\ell E) ]$, $\ell \in \{1, \dots, L\}$, and $\breve
\theta^\TRANS = [ A, B]$ denote the parameters of the dynamics of the eigen
and auxiliary systems, respectively. Then, for any policy $\pi = (\pi_1,
\pi_2, \dots)$, the performance of the eigensystem $(\ell, i)$, $\ell \in \{1,
\dots, L\}$ and $i \in N$, is given by $q^\ell J^{\ell,i}(\pi;\theta^\ell)$,
where
\begin{equation*}
  J^{\ell,i}(\pi;\theta^\ell) = \limsup_{T \to \infty} \frac {1}{T}
  \EXP^\pi\biggl[ \sum_{t=1}^{T} c(x^{\ell,i}_t, u^{\ell,i}_t)  \biggr].
\end{equation*}
Similarly, the performance of the auxiliary system~$i$, $i \in N$, is given by
$\breve q_0 \breve J^i(\pi; \breve \theta)$, where
\begin{equation*}
  \breve J^i(\pi; \breve \theta) = \limsup_{T \to \infty} \frac {1}{T}
  \EXP^\pi\biggl[ \sum_{t=1}^{T} c(\breve x^i_t, \breve u^i_t)  \biggr].
\end{equation*}
Proposition~\ref{prop:spectral} implies that the overall performance of
policy~$\pi$ can be decomposed as
\begin{equation}\label{eq:cost-split}
  J(\pi; \theta) = \sum_{i \in N} q_0 \breve J^i(\pi; \breve \theta) 
  +  \sum_{i \in N}\sum_{\ell=1}^{L}q^\ell J^{\ell,i}(\pi;\theta^\ell).
\end{equation}

The key intuition behind the result of~\cite{gao2019networked} is as
follows. By the certainty equivalence principle for LQ systems, we know that
(when the system dynamics are known) the optimal control policy of a
stochastic LQ system is the same as the optimal control policy of the
corresponding deterministic LQ system where the noises $\{w^i_t\}_{t \ge 1}$
are assumed to be zero. Note that when noises $\{w^i_t\}_{t \ge 1}$ are zero,
then the noises $\{w^{\ell,i}_t\}_{t \ge 1}$ and $\{\breve w^i_t\}_{t
\ge 1}$ of the eigen- and auxiliary-systems are also zero. This, in turn, implies that the
dynamics of all the eigen- and auxiliary systems are decoupled. These
decoupled dynamics along with the cost decoupling in~\eqref{eq:cost-split}
imply that we can choose the controls
$\{u^{\ell,i}_t\}_{t \ge 1}$ for the eigensystem $(\ell,i)$, $\ell \in
\{1,\dots, L\}$ and $i \in N$,  to
minimize\footnote{\label{fnt:scaling}The cost of the eigensystem~$(\ell,i)$ is
  $q^\ell J^{\ell,i}(\pi; \theta^\ell)$. From (A3), we know that $q^\ell$ is
  positive. Therefore, minimizing $q^\ell J^{\ell,i}(\pi; \theta^\ell)$ is the
same as minimizing $J^{\ell,i}(\pi; \theta^\ell)$.}
$J^{\ell,i}(\pi;\theta^\ell)$ and choose the controls $\{\breve u^i_t\}_{t \ge
1}$ for the auxiliary system~$i$, $i \in N$, to minimize\footnote{The same
remark as footnote~\ref{fnt:scaling} applies here.} $\breve J^i(\pi; \breve
\theta)$. These optimization problems are standard optimal control problems.
Therefore, similar to~\cite[Thoerem~3]{gao2019networked}, we obtain the
following result.
\begin{theorem}\label{thm:planning}
  Let $\brMAT S$ and $\{\MAT S^\ell\}_{\ell=1}^L$ be the solution of the
  following discrete time algebraic Riccati equations (DARE): 
  \begin{subequations}\label{eq:DARE}
    \begin{align}
      \brMAT S(\breve \theta) &= \DARE(\MAT A, \MAT B, \MAT Q, \tfrac{r_0}{q_0} \MAT R),
      \intertext{and for $\ell \in \{1,\dots,L\}$,}
      \MAT S^\ell(\theta^\ell) &= \DARE(\MAT A^\ell, \MAT B^\ell, \MAT Q,
      \tfrac{r^\ell}{q^\ell} \MAT R).
    \end{align}
  \end{subequations}
  Define the gains:
  \begin{subequations}\label{eq:gains}
    \begin{align}
      \brMAT G(\breve \theta) &=
      -\bigl( (\MAT B)^\TRANS \brMAT S(\breve \theta)\MAT B +  \tfrac{r_0}{q_0} \MAT R \bigr)^{-1}
      (\MAT B)^\TRANS \brMAT S(\breve \theta) \MAT A, \label{eq:gains1}
      \intertext{and for $\ell \in \{1,\dots,L\}$,}
      \MAT G^\ell(\theta^\ell) &=
      -\bigl( (\MAT B^{\ell})^\TRANS \MAT S^\ell(\theta^\ell) \MAT B^\ell +
      \tfrac{r^\ell}{q^\ell} \MAT R \bigr)^{-1}
      (\MAT B^{\ell})^\TRANS \MAT S^\ell(\theta^\ell) \MAT A^\ell. \label{eq:gains2}
    \end{align}
  \end{subequations}
  Then, under assumptions \textup{(A1)--(A3)}, the policy
  \begin{equation}\label{eq:optimal}
    u^i_t = \brMAT G(\breve \theta) \breve x^i_t + \sum_{\ell=1}^{L} \MAT G^\ell(\theta^\ell)  x^{\ell,i}_t
  \end{equation}
  minimizes the long-term average cost in~\eqref{eq:cost_performance} over all admissible policies.
  Furthermore, the optimal performance is given by 
  \begin{equation} \label{eq:performance}
    J(\theta) = \sum_{i \in N} q_0 \breve J^i(\breve \theta) 
    +  \sum_{i \in N}\sum_{\ell=1}^{L}q^\ell J^{\ell,i}(\theta^\ell),
  \end{equation}
  where 
  \begin{equation}\label{eq:breve-J}
    \breve J^i(\breve \theta) = (\breve v^i)^2 \TR(W \brMAT S)
  \end{equation}
 and for $\ell \in \{1,
  \dots, L\}$,
  \begin{equation}\label{eq:J-ell}
    J^{\ell,i}(\theta^\ell) = (v^{\ell,i})^2 \TR(W \MAT S^\ell).
  \end{equation}
\end{theorem}

\section{Learning for network-coupled subsystems}\label{sec:learning}
For the ease of notation, we define $
z^{\ell,i}_t = \VEC( x^{\ell,i}_t, u^{\ell,i}_t)$ and $\breve z^i_t = \VEC(\breve x^i_t,
\breve u^i_t)$. Then, we can write the
dynamics~\eqref{eq:field-dynamics}, \eqref{eq:aux-dynamics} of the eigen and
the auxiliary systems~as
\begin{subequations}\label{eq:bar-dynamics}
\begin{align}
  {x}^{\ell,i}_{t+1} &= (\theta^\ell)^\TRANS {z}^{\ell,i}_t +  w^{\ell,i}_t,
  &&\forall i \in N,
  \forall \ell \in \{1, \dots, L\},
  \label{eq:mf-bar}
  \\
  \breve{x}^i_{t+1} &= (\breve{\theta})^\TRANS \breve{z}^i_{t} +
  \breve{w}^i_{t},
  && \forall i \in N.
  \label{eq:rel-bar}
\end{align}
\end{subequations}

\subsection{Simplifying assumptions}

We impose the following assumptions to simplify the description of the algorithm and the regret analysis. 
\begin{description}
  \item[(A4)]
    The noise covariance $W$ is a scaled identity matrix given by $\sigma_w^2 I$.
  \item[(A5)] 
    For each $i \in N$, $\breve v^i \neq 0$.
\end{description}

Assumption (A4) is commonly made in most of the literature on regret analysis
of LQG systems. An implication of (A4) is that $\VAR(\breve w^i_t) = (\breve
\sigma^i)^2 I$ and $\VAR(w^{\ell,i}_t) = (\sigma^{\ell,i})^2 I$, where
\begin{equation}\label{eq:sigma}
  (\breve \sigma^i)^2 = (\breve v^i)^2 \sigma_w^2
  \quad\text{and}\quad
  (\sigma^{\ell,i})^2 = (v^{\ell,i})^2 \sigma_w^2.
\end{equation}

Assumption (A5) is made to rule out the case where the dynamics of some of the auxiliary systems are deterministic. 

\subsection{Prior and posterior beliefs:}

We assume that the unknown parameters $\breve \theta$ and $\{ \theta^\ell
\}_{\ell = 1}^L$ lie in compact subsets $\breve \Theta$ and $\{
\Theta^\ell\}_{\ell=1}^L$ of $\reals^{(d_x + d_u)
\times d_x}$. Let $\breve \theta^k$ denote the $k$-th column of
$\breve \theta$. Thus $\breve \theta = [\breve \theta^1, \dots, \breve
\theta^{d_x}]$. Similarly, let $ \theta^{\ell,k}$ denote the $k$-th
column of $ \theta^\ell$. Thus, $ \theta^\ell = [\theta^{\ell,1}, \dots,
\theta^{\ell,d_x}]$.
We use  $p\bigr|_{\Theta}$ to
denote the restriction of probability distribution $p$ on the set $\Theta$.

We assume that $\breve \theta$ and $\{ \theta^{\ell} \}_{\ell=1}^L$ are random
variables that are independent of the initial states and the noise processes.
Furthermore, we assume that the priors $\breve p_1$ and $\{p^\ell_1\}_{\ell=1}^L$ on
$\breve \theta$ and $\{ \theta^{\ell} \}_{\ell=1}^L$, respectively, satisfy
the following:
\begin{description}
  \item[(A6)] $\breve p_1$ is given as:
   \[
      \breve p_1(\breve \theta) =  \biggl[\prod_{k=1}^{d_x} \breve
    \xi^{k}_1(\breve \theta^k)\biggr]\biggr|_{\breve \Theta}  
    \]
      where for $k \in \{1, \dots, d_x\}$, 
      $\breve \xi_1^{k} = \mathcal{N}(\breve \mu_1^k, \breve \Sigma_1)$
   with mean $\breve \mu_1^k \in
   \reals^{d_x + d_u}$  and positive-definite covariance
    $\breve \Sigma_1 \in \reals^{(d_x + d_u) \times (d_x + d_u)}$.
  \item[(A7)] For each $\ell \in \{1, \dots, L\}$, $ p^\ell_1$ is given as:
 \[
       p^\ell_1( \theta^\ell) =  \biggl[\prod_{k=1}^{d_x}  {\xi}^{\ell,k}_1( \theta^{\ell, k})\biggr]\biggr|_{ \Theta^\ell}   
       \]
     where for $k \in \{1, \dots, d_x\}$, 
     $ \xi_1^{\ell,k} =
     \mathcal{N}( \mu^{\ell, k}_1,  \Sigma^\ell_1)$ with
     mean $\mu^{\ell, k}_1 \in \reals^{(d_x + d_u)}$ and
     positive-definite covariance $\Sigma^\ell_1 \in \reals^{(d_x + d_u)
     \times (d_x + d_u)}$.
\end{description}

These assumptions are similar to the assumptions on the prior 
in the recent literature on Thompson sampling for LQ systems~\cite{ouyang2017control, ouyang2019posterior}.

Our learning algorithm (and TS-based algorithms in general) keeps track of a posterior distribution on the unknown parameters based on observed data. Motivated by the nature of the planning solution (see Theorem~\ref{thm:planning}), we maintain
separate posterior distributions on $\breve \theta$ and
$\{\theta^\ell\}_{\ell=1}^L$. For each $\ell$, we select an subsystem $\iell$ such that the $\iell$-th component of the eigen-vector $v^{\ell}$ is non-zero (i.e. $v^{\Iell} \neq 0$) . At time~$t$, we maintain a posterior distribution  $p^\ell_t$ on $ \theta^\ell$
based on the corresponding eigen state and action history of the $\iell$-th subsystem. In other words, for any Borel
subset $B$ of $\reals^{(d_x + d_u) \times d_x}$,  $p^\ell_t(B)$  gives the following conditional probability
\begin{equation}
    p^\ell_t(B) = \PR( \theta^\ell \in B \mid  x^{\Iell}_{1:t},  u^{\Iell}_{1:t-1}). \label{eq:bar_posterior}
\end{equation}

We maintain a separate posterior distribution
$\breve p_t$ on $\breve \theta$ as follows. At each time $t > 1$, we select
an subsystem  $ j_{t-1} = \arg \max_{i \in N} \breve{z}^{i^\TRANS}_{t-1} \breve \Sigma_{t-1} \breve z^i_{t-1}/(\breve \sigma^i_t)^2$, where $\breve \Sigma_{t-1}$ is a
covariance matrix defined recursively  in Lemma \ref{lem:posterior} below. Then, for any Borel subset $B$ of
$\reals^{(d_x + d_u) \times d_x}$, 
\begin{equation}
  \breve p_t(B) = \PR(\breve \theta \in B \mid 
    \{ \breve x^{j_s}_s, \breve u^{j_s}_s, \breve x^{j_s}_{s+1} \}_{1 \le
      s < t }\} ), \label{eq:breve_posterior}
\end{equation}
See  \cite{gagrani2020thompson} for a discussion on the rule to select $j_{t-1}$.

\begin{lemma}\label{lem:posterior}
  The posterior distributions $p^\ell_t$, $\ell \in \{1,2,\dots,L\}$, and  $\breve p_t$ are given as follows:
  \begin{enumerate}
  \item  $p^\ell_1$ is given by Assumption (A7) and for any $t \geq 1$, 
  \[
  p^\ell_{t+1}(\theta^\ell) = \biggl[ \prod_{k=1}^{ d_x}  \xi_{t+1}^{\ell,k}(
     \theta^{\ell, k}) \biggr]\biggr|_{ \Theta^\ell},
     \]
     where for $k \in \{1, \dots, d_x\}$,  
     $ \xi_{t+1}^{\ell,k} =
      \mathcal{N}( \mu^{\ell, k}_{t+1},  \Sigma^\ell_{t+1})$, 
      and  
      
      \begin{subequations}\label{eq:p-bar}
      \begin{align}
      \mu^\ell_{t+1} &= \mu^\ell_t + 
        \frac{ 
          \Sigma^\ell_t  z^{\Iell}_t
          \bigl(  x^{\Iell}_{t+1} - (\mu^\ell_t)^\TRANS
        z^{\Iell}_t \bigr)^\TRANS
        }
        { (\sigma^{\Iell})^2 +   (z^{\Iell}_t)^\TRANS  \Sigma^\ell_t
        z^{\Iell}_t }
        , \label{eq:mu_bar_update} 
        \\
         (\Sigma^\ell_{t+1})^{-1} &=  (\Sigma^\ell_{t})^{-1}
       + \frac{1}{(\sigma^{\Iell})^2}  z^{\Iell}_t (z^{\Iell}_t)^\TRANS,
        \label{eq:sigma_bar_update}
      \end{align}
      \end{subequations}
      where, for each $t$, $\mu^\ell_{t}$ denotes the matrix $[\mu^{\ell,1}_{t}, \ldots, \mu^{\ell,d_x}_{t}]$.
      \item $\breve p_1$ is given by Assumption (A6) and for any $t \geq 1$,
       \[ 
       \breve p_{t+1}(\breve \theta) = \biggl[ \prod_{k=1}^{d_x} \breve
    \xi^{k}_{t+1}(\breve \theta^k) \biggr]\biggr|_{\breve \Theta},
    \]
      where for $k \in \{1, \dots, d_x\}$,  
      $\breve \xi^{k}_{t+1} =
    \mathcal{N}(\breve \mu_{t+1}^k, \breve \Sigma_{t+1})$,
     and 
      \begin{subequations}\label{eq:p-breve}
      \begin{align}
      \breve \mu_{t+1} &= \breve \mu_{t} + 
        \frac{ 
          \breve \Sigma_t \breve z^{j_t}_t 
          \bigl( \breve x^{j_t}_{t+1} - (\breve \mu_t)^\TRANS
        \breve z^{j_t}_t \bigr)^\TRANS }
        { (\breve{\sigma}^{j_t})^2 +  (\breve z^{j_t}_t)^\TRANS \breve \Sigma_t 
        \breve z^{j_t}_t },
        \label{eq:mu_breve_update}
        \\
        (\breve \Sigma_{t+1})^{-1} &= (\breve \Sigma_t)^{-1} 
        + \frac{1}{(\breve{\sigma}^{j_t})^2} \breve z^{j_t}_t (\breve z^{j_t}_t)^\TRANS.
        \label{eq:sigma_breve_update}
      \end{align}
      \end{subequations}
       where, for each $t$, $\breve\mu_{t}$ denotes the matrix $[\breve\mu^{1}_{t}, \ldots, \breve\mu^{d_x}_{t}]$.
  \end{enumerate}
\end{lemma}

\begin{proof}
  Note that the dynamics of ${x}^{\Iell}_t$ and $\breve{x}^i_t$ in
  \eqref{eq:bar-dynamics} are linear and the noises ${w}^{\Iell}_t$ and $\breve w^i_t$ are Gaussian. Therefore, the result follows from
  standard results in Gaussian linear regression
  \cite{sternby1977consistency}.
\end{proof}

\subsection{The Thompson sampling algorithm:}

We propose a Thompson sampling {based} algorithm called
\texttt{Net-TSDE} which is inspired by the \texttt{TSDE} (Thompson sampling
with dynamic episodes) algorithm proposed in~\cite{ouyang2017control,
ouyang2019posterior} and the structure of the optimal planning solution
described in Sec.~\ref{sec:planning}. The Thompson sampling part of our
algorithm is modeled after the modification of \texttt{TSDE} presented
in~\cite{altproof}.

The \texttt{Net-TSDE} algorithm consists of a coordinator $\mathcal C$ and
$|L|+1$ \emph{actors}: an auxiliary actor $\breve{\mathcal{A}}$ and an  eigen actor ${\mathcal A^\ell}$ for each $\ell \in \{1,2,\dots,L\}$. These
actors are described below and the whole algorithm is presented in
Algorithm~\ref{alg:tsde_mf}.
\begin{itemize}
  \item At each time, the coordinator $\mathcal C$ observes the current global
    state $x_t$, computes the eigenstates $\{ x^\ell_t\}_{\ell =1}^L$ and
    the auxiliary states $\breve x_t$, and sends 
    the eigenstate $ x^\ell_t$ to  the eigen actor
    ${\mathcal A^\ell}$, $\ell \in \{1, \dots, L\}$, and sends the auxiliary state $\breve x_t$ to the
    auxiliary actor~$\breve{\mathcal A}$. The eigen actor
    ${\mathcal A^\ell}$, $\ell \in \{1, \dots, L\}$, computes the eigencontrol $ u^\ell_t$ and the
    auxiliary actor $\breve {\mathcal
    A}$ computes the auxiliary control $\breve u_t$
    (as per the details presented below) and both send their computed controls back to the
    coordinator~$\mathcal{C}$. The coordinator then computes and executes the
    control action $u^i_t = \sum_{\ell=1}^{L} u^{\ell,i}_t + \breve u^i_t$ for each subsystem~${i
    \in N}$ .

  \item The eigen actor ${\mathcal{A}^\ell}$, $\ell \in \{1, \dots, L\}$, maintains the posterior $p^\ell_t$ on $ \theta^\ell$ according to~\eqref{eq:p-bar}. The actor works in
    episodes of dynamic length. Let $t^\ell_k$ and $T^\ell_k$ denote the starting time
    and the length of episode~$k$, respectively. Each episode is of a minimum length $T^\ell_{\min} +1$, where $T^\ell_{\min}$ is chosen as described in~\cite{altproof}. Episode~$k$ ends if the
    determinant of covariance $ \Sigma^\ell_t$ falls below  half of its value at
   time \blue{$t^\ell_k$  (i.e., $\det( \Sigma^\ell_t) < \tfrac12 \det \Sigma_{t^\ell_k }$)} or if the length of the episode is one more than the
    length of the previous episode (i.e., $t - t^\ell_k >  T^\ell_{k-1}$).
    Thus,
    \begin{align*}\label{eq:stopping1}
    t^\ell_{k+1} = \min \left\{ t > t^\ell_k + T^\ell_{\min} 
      \,\middle|\, \begin{lgathered}
        t - t^\ell_k > T^\ell_{k-1} \text{ or } \\
      \det \Sigma^\ell_t < \blue{\tfrac12 \det \Sigma_{t^\ell_k} }
  \end{lgathered} \right\}
\end{align*}
    At the beginning of episode~$k$, the eigen actor ${\mathcal{A}^\ell}$
    samples a parameter $ \theta^\ell_k$ according to the posterior distribution $
    p^\ell_{t^\ell_k}$. During episode~$k$, the eigen actor ${\mathcal{A}^\ell}$
    generates the eigen controls using   the sampled  parameter $ \theta^\ell_k$, 
    i.e., $ u^\ell_t = \MAT G^\ell( \theta^\ell_k)  x^\ell_t$.
  \item The auxiliary actor $\breve{\mathcal A}$ is similar to the
    eigen actor. Actor $\breve{\mathcal A}$ maintains the posterior
    $\breve p_t$ on $\breve \theta$ according to~\eqref{eq:p-breve}. The
    actor works in episodes of dynamic length. The episodes of the auxiliary
    actor $\breve{\mathcal{A}}$ and the eigen actors ${\mathcal{A}^\ell}$, $\ell \in \{1,2,\dots,L\}$,
    are separate from each other.\footnote{The
      episode count~$k$ is used as a local variable for each
    actor.} Let $\breve t_k$ and $\breve T_k$ denote
    the starting time and the length of episode~$k$, respectively. Each episode is of a minimum length $\breve T_{\min}+1$, where $\breve T_{\min}$ is chosen as described in~\cite{altproof}. The termination
    condition for each episode is similar to that of the eigen
    actor~${\mathcal{A}^\ell}$. In particular,
    \begin{align*}\label{eq:stopping}
    \breve t_{k+1} = \min \left\{ t > \breve t_k + \breve T_{\min} 
      \,\middle|\, \begin{lgathered}
        t - \breve t_k >\breve T_{k-1} \text{ or } \\
      \det \breve \Sigma_t <\blue{\tfrac12 \det \breve \Sigma_{\breve t_k}}  
  \end{lgathered} \right\}
\end{align*}
    At the beginning of episode~$k$, the auxillary actor $\breve{\mathcal{A}}$
    samples a parameter $\breve \theta_k$ from the posterior distribution $\breve
    p_{\breve t_k}$. During episode~$k$, the auxiliary actor
    $\breve{\mathcal{A}}$ generates the auxiliary controls using the   the sampled  parameter
    $\breve \theta_k$, i.e., $\breve u_t 
     = \brMAT G(\breve \theta_{k})\breve x_t$.
\end{itemize}

\begin{algorithm}[!t]
\caption{\texttt{Net-TSDE}}
\label{alg:tsde_mf}
\begin{algorithmic}[1]
  \State \textbf{initialize eigen actor:} $ \Theta^\ell$, $( \mu^\ell_1,
  \Sigma^\ell_1)$, $ t^\ell_0 = -T_{\min}$, $ T^\ell_{-1} = T_{\min}$, $k = 0$, $\theta^\ell_k = 0$
  \State \textbf{initialize auxiliary actor:} $\breve \Theta$, $(\breve
  \mu_1, \breve \Sigma_1)$, $\breve t_0 = -T_{\min}$, $\breve T_{-1} = T_{\min}$, $k = 0$, $\breve \theta_k = 0$.
\For{$t = 1, 2, \dots $}
  \State observe $x_t$
  \State compute $\{ x^\ell_t \}_{\ell=1}^L$ and $\breve x_t$ using~\eqref{eq:xell} and~\eqref{eq:xbreve}.
  \For{$\ell = 1, 2, \dots, L $}
  \State $u^\ell_t \gets \textsc{eigen-actor}(x^\ell_t)$
  \EndFor
      \State $\breve u_t \gets \textsc{auxiliary-actor}(\breve x_t)$
      \For{$i \in N$}
      \State Subsystem~$i$ applies control $u^i_t =  u^{\ell,i}_t + \breve
      u^i_t$
      \EndFor
 \EndFor
\end{algorithmic}
\medskip
\begin{algorithmic}[1]
\Function{eigen-actor}{$ x^\ell_t$}
  \State \textbf{global var} $t$
  \State Update $ p^\ell_t$ according~\eqref{eq:p-bar}
  \State \textbf{if} $(t - t^\ell_k > T_{\min})$ and 
  \State \quad ($(t - t^\ell_k > T^\ell_{k-1})$ or 
    $(\det \Sigma^\ell_t < \blue{\tfrac12 \det \Sigma_{t^\ell_k}})$)
  \State  \textbf{then}
  \State \qquad $T^\ell_k \gets t - t^\ell_k$,
          $k \gets k + 1$,
          $ t^\ell_k \gets t$
  \State \qquad sample $ \theta^\ell_k \sim  p^\ell_t$
\State \textbf{return} $\MAT G^\ell(\theta^\ell_k)  x^\ell_t$
\EndFunction
\end{algorithmic}
\medskip
\begin{algorithmic}[1]
  \Function{auxiliary-actor}{$\breve x_t$}
  \State \textbf{global var} $t$
  \State Update $\breve p_t$ according~\eqref{eq:p-breve}
  \State \textbf{if} $(t - \breve t_k > T_{\min})$ and 
  \State \quad ($(t - \breve t_k > \breve T_{k-1})$ or 
    $(\det \breve \Sigma_t < \blue{\tfrac12 \det \breve \Sigma_{t^\ell_k}})$)
  \State  \textbf{then}
    \State \qquad $\breve T_k \gets t - \breve t_k$,
           $k \gets k + 1$,
           $\breve t_k \gets t$
    \State \qquad sample $\breve \theta_k \sim \breve p_t$
\State \textbf{return} $\brMAT G(\breve \theta_k)  \breve x_t$
\EndFunction

\end{algorithmic}
\end{algorithm}

Note that the algorithm does not depend on the horizon~$T$.

\subsection{Regret bounds:}

We impose the following assumption to ensure that
the closed loop dynamics of the eigenstates and the auxiliary states of
each subsystem are stable. 

\begin{description}
\item[(A8)] There exists a positive number $\delta \in (0, 1)$ such that 
  \begin{itemize}
  \item for any  $\ell \in \{1,2,\dots,L\}$ and $ \theta^\ell,  \phi^\ell \in  \Theta^\ell$ where $
    (\theta^{\ell})^\TRANS = [\MAT A^\ell_{ \theta^\ell}, \MAT B^\ell_{ \theta^\ell} ]$, we
        have
        \[
      \rho(A^\ell_{\theta^\ell} + B^\ell_{\theta^\ell} G^\ell(\phi^\ell)) \le \delta.
    \]
  \item for any 
        $\breve \theta,
        \breve \phi \in \breve \Theta$, 
        where $(\breve \theta)^\TRANS = [\MAT A_{\breve \theta}, \MAT B_{\breve
        \theta}]$, we have
    \[
      \rho(A_{\breve \theta} + B_{\breve \theta} \breve G(\breve \phi)) \le \delta.
    \]
  \end{itemize}

\end{description}

This assumption is similar to an assumption made in \cite{altproof} for TS
for LQG systems. According to~\cite[Lemma~1]{faradonbeh2019finite} (also
see~\cite[Theorem~11]{simchowitz2020naive}), (A8) is satisfied if
\begin{align*}
  \Theta^\ell &= \{ \theta^\ell \in \reals^{(d_x + d_u)\times d_x} :
  \| \theta^\ell - \theta^\ell_\circ \| \le \varepsilon^\ell \},
  \\
  \breve \Theta &= \{ \breve \theta \in \reals^{(d_x + d_u)\times d_x} :
  \| \breve \theta - \breve \theta_\circ \| \le \breve \varepsilon \},
\end{align*}
where $\theta^\ell$ and $\breve \theta$ are stabilizable and
$\varepsilon^\ell$ and $\breve \varepsilon$ are sufficiently small.
In other words, the assumption holds when the true system is in a small
neighborhood of a known nominal system. Such a the small neighborhood can be
learned with high probability by running appropriate stabilizing procedures
for finite time~\cite{faradonbeh2019finite, simchowitz2020naive}.
 
The following result provides an 
upper bound on the regret of the proposed algorithm.
\begin{theorem}\label{thm:main}
  Under \textup{(A1)--(A8)}, the regret of \textup{\texttt{Net-TSDE}} is upper
  bounded as follows:
  \begin{equation*}
    R(T;\textup{\texttt{Net-TSDE}} ) \leq \tilde{\mathcal O} \bigl(
      \alpha^{\mathcal{G}}
       \sigma_w^2
     d_x^{0.5}(d_x + d_u) \sqrt{T} \bigr),
  \end{equation*}
  where
  \(
    \alpha^{\mathcal{G}} = \sum_{\ell=1}^{L} q^\ell
    + q_0(n-L).
  \)
\end{theorem}
 See Section~\ref{sec:proof} for proof.

\begin{remark}
The term $\alpha^{\mathcal{G}}$ in the regret bound partially captures the impact of the network on the regret. The coefficients $r_0$ and $\{r^\ell\}_{\ell =1}^L$ depend on the network and also affect the regret but their dependence is hidden inside the $\tildeO(\cdot)$ notation. It is possible to explicitly characterize this dependence but doing so does not provide any additional insights. We discuss the impact of the network coupling on the regret in Section \ref{sec:examples} via some examples.
\end{remark}

\begin{remark}
  The regret per subsystem is given by $R(T;\textup{\texttt{Net-TSDE}})/n$, which is proportional to 
\[
  \alpha^{\mathcal{G}}/n = \mathcal{O}\Bigl(\frac{L}{n}\Bigr)
  + \mathcal{O}\Bigl(\frac{n-1}{n}\Bigr) 
  = \mathcal{O}\Bigl(1 + \frac{L}{n}\Bigr).
\]
Thus, the regret per-subsystem scales as $\mathcal{O}(1+L/n)$. In contrast, for the
standard \texttt{TSDE}
algorithm~\cite{ouyang2017control,ouyang2019posterior,altproof}, the regret
per subsystem is proportional to
\(
  \alpha^{\mathcal{G}}({\tt TSDE})/n = \mathcal{O}(n^{0.5}).
\)
This clearly illustrates the benefit of the proposed learning algorithm.
\end{remark}

\section{Regret analysis}\label{sec:proof}
For the ease of notation, we simply use $R(T)$ instead of $R(T;\texttt{Net-TSDE})$
in this section. Based on Proposition~\ref{prop:spectral}
and~Theorem~\ref{thm:planning}, the regret may be decomposed as
\begin{equation}\label{eq:regret-split}
  R(T) =  
  \sum_{i \in N} q_0\breve R^{i}(T)
  +
  \sum_{i \in N}\sum_{\ell=1}^{L}q^\ell R^{\ell,i}(T) 
\end{equation}
where 
\begin{align*}
  \breve{R}^{i}(T) &:= \EXP\biggl[ \sum_{t=1}^T
    \breve{c}(\breve{x}_t^{i},\breve{u}_t^{i} ) - T
    \breve{J}^i(\breve{\theta})
  \biggr], 
  \\
  \intertext{and, for $\ell \in \{1, \dots, L\}$,}
  R^{\ell,i}(T) &:= \EXP \biggl[ \sum_{t=1}^T  {c}^\ell({x}^{\ell,i}_t,{u}^{\ell,i}_t) -
  T {J}^{\ell,i}({\theta^{\ell}}) \biggr].
\end{align*}

Based on the discussion at the beginning of Sec.~\ref{sec:planning}, $\breve
q_0 \breve R^i(T)$, $i \in N$, is the regret associated with auxiliary
system~$i$ and $q^\ell R^{\ell,i}(T)$, $\ell \in \{1,\dots, L\}$ and $i \in
N$, is the regret associated with eigensystem~$(\ell,i)$. We now bound $\breve
R^i(T)$ and $R^{\ell,i}(T)$ separately.

\subsection{Bound on $R^{\ell,i}(T)$}

Fix $\ell \in \{1,\dots, L\}$. For the component~$\iell$,
the \texttt{Net-TSDE} algorithm is exactly same as the variation of the
\texttt{TSDE} algorithm of~\cite{ouyang2019posterior} presented
in~\cite{altproof}. Therefore, from~\cite[Theorem~1]{altproof}, it follows
that
\begin{equation}\label{eq:R-iell}
  R^{\Iell}(T) \le \tildeO\bigl( (\sigma^{\Iell})^2 d_x^{0.5}(d_x + d_u) \sqrt{T}) 
  \bigr).
\end{equation}
We now show that the regret of other eigensystems $(\ell,i)$ with $i \neq
\iell$ also satisfies a similar bound.

\begin{lemma}\label{lem:eigen-regret}
  The regret of eigensystem $(\ell, i)$, $\ell \in \{1, \dots, L\}$ and $i \in
  N$, is bounded as follows:
  \begin{equation}\label{eq:R-ell}
    R^{\ell,i}(T) \leq  
    \tildeO\bigl( (\sigma^{\ell,i})^2  d_x^{0.5} (d_x+d_u) \sqrt{T} 
    \bigr).
  \end{equation}
\end{lemma}
\begin{proof}
  Fix $\ell \in \{1,\dots, L\}$. 
  Recall from~\eqref{eq:xell} that $x^\ell_t = x_t v^\ell (v^\ell)^\TRANS$.
  Therefore, for any $i \in N$, 
  \[
    x^{\ell,i}_t = x_t v^\ell v^{\ell,i} = v^{\ell,i} x_t v^\ell  ,
  \]
  where the last equality follows because $v^{\ell,i}$ is a scalar. 
  Since we are using the same gain $G^\ell(\theta^\ell_k)$ for all agents $i
  \in N$, we have
  \[
    u^{\ell,i}_t = G^\ell(\theta^\ell_k) x^{\ell,i}_t =
    v^{\ell,i} G^{\ell}(\theta^\ell_k) x_t v^\ell.
  \]
  Thus, we can write (recall that $\iell$ is chosen such that $v^{\Iell} \neq
  0$),
  \[
    x^{\ell,i}_t = \Bigl( \frac{v^{\ell,i}}{v^{\Iell}} \Bigr) x^{\Iell}_t
    \text{ and }
    u^{\ell,i}_t = \Bigl( \frac{v^{\ell,i}}{v^{\Iell}} \Bigr) u^{\Iell}_t,
    \quad \forall i \in N.
  \]
  Thus, for any $i \in N$, 
  \begin{equation}\label{eq:cost-prop}
    c^\ell(x^{\ell,i}_t, u^{\ell,i}_t) = 
    \Bigl( \frac{v^{\ell,i}}{v^{\Iell}} \Bigr)^2
    c^\ell(x^{\Iell}_t, u^{\Iell}_t).
  \end{equation}
  Moreover, from~\eqref{eq:J-ell}, we have
  \begin{equation}\label{eq:performance-prop}
    J^{\ell,i}(\theta^\ell) = 
    \Bigl( \frac{v^{\ell,i}}{v^{\Iell}} \Bigr)^2 J^{\Iell}(\theta^\ell).
  \end{equation}
  By combining~\eqref{eq:cost-prop} and~\eqref{eq:performance-prop}, we get
  \[
    R^{\ell,i}(T) = \Bigl( \frac{v^{\ell,i}}{v^{\Iell}} \Bigr)^2 R^{\Iell}(T).
  \]
  Substituting the bound for $R^{\Iell}(T)$ from~\eqref{eq:R-iell} and observing that
  $(v^{\ell,i}/v^{\Iell})^2 = (\sigma^{\ell,i}/\sigma^{\Iell})^2$ gives the result.
\end{proof}

\subsection{Bound on $\breve R^{i}(T)$}

The update of the posterior $\breve p_t$ on $\breve \theta$ does not depend on
the history of states and actions of any fixed agent~$i$. Therefore, we
cannot directly use the argument presented in~\cite{altproof} to bound the
regret $\breve R^i(T)$. We present a bound from first principles below.

For the ease of notation, for any episode $k$, we use $\brMAT G_k$ and $\brMAT
S_k$ to denote $\brMAT G(\breve \theta_k)$ and $\brMAT S(\breve \theta_k)$ respectively. 
From LQ optimal control theory~\cite{kumar2015stochastic}, we know  that the average cost $\breve J^i(\breve \theta_k)$ and the optimal  policy $\breve u^i_t = \brMAT G_k \breve x^i_t$ for the model parameter $\breve \theta_k$ satisfy the following Bellman equation:
\begin{multline*}
  \breve J^i(\breve \theta_k) + (\breve x^i_t)^\TRANS \brMAT S_k \breve x^i_t =
  \breve c(\breve x^i_t, \breve u^i_t)  \\ + 
  \EXP\bigl[ 
    \bigl(\breve \theta_k^\TRANS \breve z^i_t + \breve w^i_t\bigr)^\TRANS
    \brMAT S_k \bigl(\breve \theta_k^\TRANS \breve z^i_t + \breve w^i_t\bigr) 
  \bigr].
\end{multline*}
Adding and subtracting 
\(  
  \EXP[ (\breve x^i_{t+1})^\TRANS \brMAT S_k \breve x^i_{t+1} \mid \breve z^i_t] 
\)
and noting that $\breve x^i_{t+1} = \breve \theta^\TRANS \breve z^i_t + \breve
w^i_t$, we get that
\begin{align}
   \breve c & (\breve x^i_t, \breve u^i_t) =
  \breve J^i(\breve \theta_k) 
   + (\breve x^i_t)^\TRANS \brMAT S_k \breve x^i_t 
  -
  \EXP[ (\breve x^i_{t+1})^\TRANS \brMAT S_k \breve x^i_{t+1} | \breve z^i_t] 
  \notag \\
&+ ( \breve \theta^\TRANS \breve z^i_t)^\TRANS \brMAT S_k ((\breve
\theta)^\TRANS
  \breve z^i_t)
  - ( \breve \theta_k^\TRANS \breve z^i_t)^\TRANS \brMAT S_k ((\breve
  \theta_k)^\TRANS \breve z^i_t).
  \label{eq:bellman_breve}
\end{align}
Let $\breve K_T$ denote the number of episodes of the auxiliary actor until
horizon $T$. For each $k > \breve K_T$, we define $\breve t_k$ to be $T+1$. Then, using~\eqref{eq:bellman_breve}, we have that for any
agent~$i$,
\begin{align}
  \breve R^{i}&(T) =
  \underbrace{%
    \EXP\biggl[ \sum_{k=1}^{\breve K_T} \breve T_k \breve J^i(\breve \theta_k) - T 
    \breve J^i(\breve \theta) \biggr]
  }_{\text{regret due to sampling error\,} \eqqcolon \breve R^{i}_0(T)}
  \displaybreak[1]
  \notag\\
  & +
  \underbrace{%
   \EXP\biggl[ 
    \sum_{k=1}^{\breve K_T} \sum_{t = \breve t_k}^{\breve t_{k+1} - 1}
    \bigl[ (\breve x^i_t)^\TRANS \brMAT S_k \breve x^i_t -
    (\breve x^i_{t+1})^\TRANS \brMAT S_k \breve x^i_{t+1} \bigr]
   \biggr]
   }_{\text{regret due to time-varying controller\,} \eqqcolon \breve R^{i}_1(T)}
   \displaybreak[1]
  \notag \\
  & + 
  \underbrace{%
    \begin{aligned}[t]
   \EXP\biggl[ 
    \sum_{k=1}^{\breve K_T} \sum_{t = \breve t_k}^{\breve t_{k+1} - 1}
    \bigl[ ( \breve \theta^\TRANS \breve z^i_t)^\TRANS & \brMAT S_k ((\breve
      \theta)^\TRANS \breve z^i_t) 
      \\[-15pt] -
   &( \breve \theta_k^\TRANS \breve z^i_t)^\TRANS \brMAT S_k ((\breve
 \theta_k)^\TRANS \breve z^i_t) \bigr]
  \biggr].
   \end{aligned}}_{\text{regret due to model mismatch\,} \eqqcolon \breve R^{i}_2(T)}
   \label{eq:regret-components}
\end{align}

\begin{lemma}\label{lem:breve_regret_terms}\label{lem:aux-regret}
  The terms in~\eqref{eq:regret-components} are bounded as follows:
  \begin{enumerate}
    \item $\breve R^{i}_0(T) \le 
    \tilde{\mathcal{O}} ((\breve \sigma^i)^2 (d_x + d_u)^{0.5} \sqrt{T})$.
    \item $\breve R^{i}_1(T) \le 
    \tilde{\mathcal{O}} ((\breve \sigma^i)^2 (d_x + d_u)^{0.5} \sqrt{T})$.
    \item $\breve R^{i}_2(T) \le 
      \tilde{\mathcal{O}} ((\breve \sigma^i)^2 d_x^{0.5} (d_x + d_u) \sqrt{T})$.
  \end{enumerate}
  Combining these three, we get that
  \begin{equation}\label{eq:R-breve}
    \breve R^i(T) \le
    \tilde{\mathcal{O}} ((\breve \sigma^i)^2 d_x^{0.5} (d_x + d_u) \sqrt{T}).
  \end{equation}
\end{lemma}
See Appendix for the proof.

\subsection{Proof of Theorem~\ref{thm:main}}

For ease of notation, let $R^* = \tilde{\mathcal{O}} (d_x^{0.5} (d_x + d_u)
\sqrt{T})$. Then, by subsituting the result of
Lemmas~\ref{lem:eigen-regret} and~\ref{lem:aux-regret}
in~\eqref{eq:regret-split}, we get that
\begin{align}
  R(T) &\le \sum_{i \in N} q_0 (\breve \sigma^i)^2 R^*
  + \sum_{i \in N} \sum_{\ell = 1}^L q^\ell (\sigma^{\ell,i})^2 R^*
  \notag \\
  &\stackrel{(a)}= \sum_{i \in N} q_0 (\breve v^i)^2 \sigma_w^2 R^*
  + \sum_{i \in N} \sum_{\ell = 1}^L q^\ell (v^{\ell,i})^2 \sigma_w^2 R^*
  \notag \\
  &\stackrel{(b)}= \Bigl( \textstyle q_0 (n-L) + \sum\limits_{\ell=1}^L q^\ell \Bigr)
     \sigma_w^2 R^*,
  \label{eq:regret-bound-proof}
\end{align}
where $(a)$ follows from~\eqref{eq:sigma} and $(b)$ follows from observing
that $\sum_{i \in N} (v^{\ell,i})^2 = 1$ and therefore $\sum_{i \in N} (\breve
v^i)^2 = n - L$. Eq.~\eqref{eq:regret-bound-proof} establishes the result of
Theorem~\ref{thm:main}.

\section{Some examples}\label{sec:examples}

\subsection{Mean-field system}\label{sec:mf}

Consider a complete graph $\mathcal{G}$  where the edge weights are equal to $1/n$. Let $M$ be equal to the adjacency matrix of the graph, i.e., $M =  \tfrac{1}{n}\mathds{1}_{n \times n}$. Thus, the system dynamics are given by  
\[
  x^i_{t+1} = A x^i_t + B u^i_t +  D \bar x_t +  E \bar u_t + w^i_t,
\]
where $\bar x_t = \tfrac{1}{n} \sum_{i \in N} x^i_t$ and $\bar u_t = \tfrac{1}{n} \sum_{i \in N} u^i_t$.
Suppose $K_x = K_u = 1$ and $q_0 = r_0 = 1/n$ and $q_1 = r_1 = \kappa/n$, where $\kappa$ is a positive constant. 

In this case, $M$ has rank $L = 1$,  the non-zero eigenvalue of $M$ is $\lambda^1 = 1$, the corresponding normalized eigenvector is $\tfrac{1}{\sqrt{n}} \mathds{1}_{n \times 1}$ and $q^1 = r^1 = q_0 + q_1 = (1 + \kappa)/n$.
The eigenstate is given by $x^1_t = [\bar x_t, \dots, \bar x_t]$ and a similar structure holds for the eigencontrol $u^1_t$. 
The per-step cost can be written as (see Proposition~\ref{prop:spectral})
\begin{align*}
  &c(x_t, u_t) = 
     (1 + \kappa) \bigl[
    \bar x_t^\TRANS Q \bar x_t + \bar u_t^\TRANS R \bar u_t \bigr].
    \\
  & \quad  +
  \frac1n \sum_{i \in N}\bigl[ 
    (x^i_t - \bar x_t)^\TRANS Q (x^i_t - \bar x_t) + 
    (u^i_t - \bar u_t)^\TRANS R (u^i_t - \bar u_t) \bigr]
\end{align*}
Thus, the system is similar to the mean-field team system investigated in~\cite{arabneydi2015mft, arabneydi2016mft}. 

For this model, the network dependent constant $\alpha^{\mathcal{G}}$ in the regret bound of Theorem~\ref{thm:main} is given by 
\(
  \alpha^{\mathcal{G}} = \bigl(1+\frac{\kappa}{n} \bigr)
   = \mathcal{O}\bigl( 1 + \frac1n \bigr).
\)
Thus, for the mean-field system, the regret of \texttt{Net-TSDE} scales as
$\mathcal{O}(1 + \tfrac1n)$ with the number of agents. This is consistent with the discussion following Theorem~\ref{thm:main}.

We test these conclusions via numerical simulations of a scalar mean-field model with $d_x = d_u = 1$, $\sigma_w^2 = 1$, $A = 1$, $B = 0.3$, $D=0.5$, $E=0.2$, $Q=1$, $R=1$, and $\kappa = 0.5$. The uncertain sets are chosen as 
\( \Theta^1 = \{ \theta^1 \in \reals^2 : A+D + (B+E)G^1(\theta^1) < \delta \} \)
and
\( \breve \Theta = \{ \breve \theta \in \reals^2: A + B \breve G(\breve \theta) < \delta \} \)
where $\delta = 0.99$. The prior over these uncertain sets is chosen according to (A6)--(A7) where $\breve \mu_1 = \mu^1_1 = [1, 1]^\TRANS$ and $\breve \Sigma_1 = \Sigma^1_1 = I$. We set $T_{\min} = 0$ in \texttt{Net-TSDE}.
The system is simulated for a horizon of $T = 5000$ and the expected regret $R(T)$ averaged over $500$ sample trajectories is shown in Fig.~\ref{fig:mean-field}. As expected, the regret scales as $\tildeO(\sqrt{T})$ with time and $\mathcal{O}\bigl(1 + \frac1n\bigr)$ with the number of agents.

\begin{figure}[!t]
  \centering
  \begin{subfigure}[t]{0.49\linewidth}
    \includegraphics[width=\textwidth]{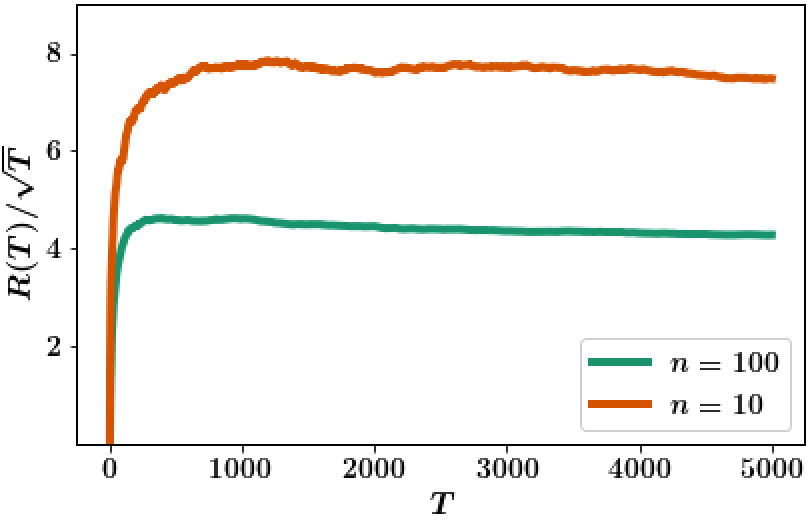}
    \caption{$R(T)/\sqrt{T}$ vs $T$}
  \end{subfigure}
  \hfill
  \begin{subfigure}[t]{0.49\linewidth}
    \includegraphics[width=\textwidth]{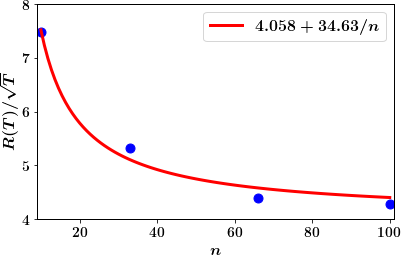}
    \caption{$R(T)/\sqrt{T}$  vs number of agents.}
  \end{subfigure}
  \caption{Regret for mean-field system.}
  \label{fig:mean-field}
\end{figure}

\begin{delete}
  \blue{
Note that the coupling graph $M$ in this model is dense. Nonetheless, the regret-per agent remains almost constant with the number of agents. This is in contrast to the scaling results for factored MDPs~\cite{NIPS2014_0deb1c54, chen2021efficient}, where regret scaling is established under the assumption of sparse interaction between agents. 
}
\end{delete}

\subsection{A general low-rank network}

\begin{figure}[!htb]
  \centering
\begin{tikzpicture}[thick,scale=0.9]
    \node [agent] at (0, 0) (3) {$3$};
    \node [agent] at (1, 0) (4) {$4$}; 
    \node [agent] at (0, 1) (2) {$2$};
    \node [agent] at (1, 1) (1) {$1$};

    \path (1) edge node[above] {$a$} (2)
          (2) edge node[left]  {$a$} (3)
          (3) edge node[below] {$b$} (4)
          (4) edge node[right] {$b$} (1);

    \node at (3.5, 0.5) {$\begin{bmatrix}
           0 & a & 0 & b \\
           a & 0 & a & 0 \\
           0 & a & 0 & b \\
           b & 0 & b & 0 
          \end{bmatrix}$};
  \end{tikzpicture}
  \caption{Graph $\mathcal{G}^\circ$ with $n=4$ nodes and its adjacency matrix}
  \label{fig:graph}
\end{figure}
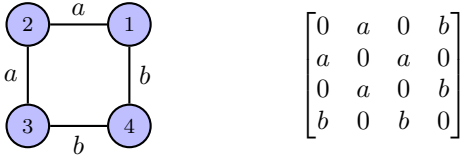

We consider a network with $4n$ nodes given by the graph $\mathcal{G} = \mathcal{G}^\circ \otimes \mathcal{C}_n$, where $\mathcal{G}^\circ$ is a 4-node graph shown in Fig.~\ref{fig:graph} and $\mathcal{C}_n$ is the complete graph with $n$ nodes and each edge weight equal to~$\frac{1}{n}$. Let $M$ be the adjacency matrix of $\mathcal{G}$  which is given as  $M = M^\circ \otimes \frac1n \mathds{1}_{n\times n}$, where $M^\circ$ is the adjacency matrix of $\mathcal{G}^\circ$ shown in Fig.~\ref{fig:graph}. Moreover, suppose $K_x = 2$ with $q_0 = 1$, $q_1 = -2$, and $q_2 = 1$ and $K_u = 0$ with $r_0 = 1$.  Note that the cost is not normalized per-agent.

In this case, the rank of $M^\circ$ is $2$ with eigenvalues $\pm \rho$, where $\rho = \sqrt{2(a^2 + b^2)}$ and the rank of $\frac1n \mathds{1}_{n\times n}$ is $1$ with eigenvalue $1$. Thus, $M = M^\circ \otimes \frac1n \mathds{1}_{n\times n}$ has the same non-zero eigenvalues as $M^\circ$ given by $\lambda^1 = \rho$ and $\lambda^2 = -\rho$. Further, $q^\ell = (1 - \lambda^\ell)^2$ and $r^\ell = 1$, for $\ell \in \{1, 2\}$. We assume that $a^2 + b^2 \neq 0.5$, so that the model satisfies (A3). 

For this model, the scaling parameter $\alpha^{\mathcal{G}}$ in the regret bound in Theorem~\ref{thm:main} is given by 
\[
  \alpha^{\mathcal{G}} = (1 - \rho)^2 + (1 + \rho)^2 + (4n - 2)
  = 4n + 2\rho^2.
\]
Recall that $\rho^2 = (\lambda^1)^2 = (\lambda^2)^2$. Thus, $\alpha^{\mathcal{G}}$ has an explicit dependence on the square of the eigenvalues and the number of nodes. 

\begin{figure}[!t]
  \centering
  \begin{subfigure}[t]{0.49\linewidth}
    \includegraphics[width=\textwidth]{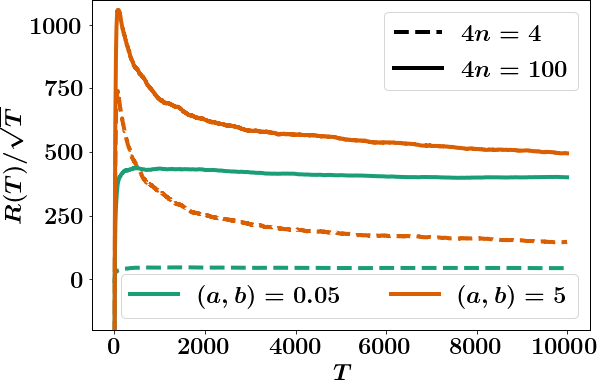}
    \caption{$R(T)/\sqrt{T}$ vs $T$}
  \end{subfigure}
  \hfill
  \begin{subfigure}[t]{0.49\linewidth}
    \includegraphics[width=\textwidth]{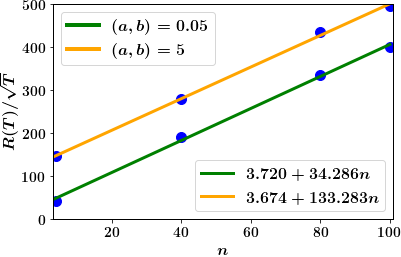}
    \caption{$R(T)/\sqrt{T}$  vs number of agents.}
  \end{subfigure}
  \caption{Regret for general low-rank network.}
  \label{fig:scaling}
\end{figure}

We verify this relationship via numerical simulations. We consider the graph above with two choices of parameters $(a,b)$: (i) $a = b = 0.05$ and (ii) $a = b = 5$. 
For both cases, we consider a scalar system with parameters same as the
mean-field system considered in Sec.~\ref{sec:mf}. The regret for both
cases with different choices of number of agents $4n \in \{4, 40, 80, 100
\}$ is shown in Fig.~\ref{fig:scaling}. As expected, the regret scales as $\tildeO(\sqrt{T})$ with time and $\mathcal{O}\bigl(4n + 2\rho^2)$ with the number of agents.

\section{Conclusion} \label{sec:conclusion}
We consider the problem of controlling   an unknown LQG system consisting of multiple subsystems connected over a network.
 By utilizing a spectral decomposition technique, we decompose the coupled subsystems into eigen and auxiliary systems.
We propose a TS-based learning algorithm \texttt{Net-TSDE} which  maintains separate posterior distributions on the unknown parameters 
$\theta^\ell, \ell \in \{1,\ldots, L\}$, and $\breve \theta$ associated with the eigen and auxiliary systems respectively.
For each eigen-system, \texttt{Net-TSDE} learns the unknown  parameter $\theta^\ell$ and controls the system in a  manner similar to the \texttt{TSDE} algorithm for single agent LQG systems proposed in~\cite{ouyang2017control, ouyang2019posterior,altproof}. Consequently, the regret for each eigen system can be bounded using the results of~\cite{ouyang2017control,ouyang2019posterior,altproof}.
However, the part of the \texttt{Net-TSDE} algorithm that performs  learning and control for the auxiliary system has an agent selection step and thus requires additional analysis to bound its regret. Combining the regret bounds for the eigen and auxiliary systems 
 shows that the total expected regret of \texttt{Net-TSDE} is upper bounded by $\tilde{\mathcal{O}}( n d_x^{0.5} ({d_x + d_u}) \sqrt{T})$.
The empirically observed scaling of  regret with respect to the time horizon $T$ and  the number of subsystems $n$ in our numerical experiments agrees with the theoretical upper bound.

The results presented in this paper rely on the spectral decomposition developed in~\cite{gao2019networked}. A limitation of this decomposition is that the local dynamics (i.e., the $(A,B)$ matrices) are assumed to be identical for all subsystems. Interesting generalizations overcoming this limitation include settings where (i)~there are multiple types of subsystems and the $(A,B)$ matrices are the same for subsystems of the same type but different across types; and (ii)~the subsystems are not identical but approximately identical, i.e., there are nominal dynamics $(A^\circ, B^\circ)$ and the local dynamics $(A^i,B^i)$ of subsystem~$i$ are in a small neighborhood of $(A^\circ, B^\circ)$.

The decomposition in~\cite{gao2019networked} exploits the fact that the dynamics and the cost couplings have the same spectrum (i.e., the same orthonormal eigenvectors). It is also possible to consider learning algorithms which exploit other features of the network such as sparsity in the case of networked MDPs~\cite{NIPS2014_0deb1c54, chen2021efficient}.

\bibliographystyle{ieeetr}
\bibliography{IEEEabrv, ref_learning}

\appendix
\section{Appendix: Regret Analysis}

\subsection{Preliminary Results}

Since $\breve{S}(\cdot)$ and $\breve{G}(\cdot)$ are continuous functions on a compact set $\breve{\Theta}$, there exist finite constants  $\breve{M}_J, \breve{M}_{\breve{\theta}}, \breve{M}_S, \breve{M}_G$ such that $\TR(\breve{S}(\breve{\theta})) \leq \breve{M}_J, \|\breve{\theta} \| \leq \breve{M}_{\breve{\theta}}, \|\breve{S}(\breve{\theta}) \| \leq \breve{M}_S$ and $\|[I, \breve{G}(\breve{\theta})^\TRANS]^\TRANS \| \leq \breve{M}_G$ for all $\breve{\theta} \in \breve{\Theta}$ where $\|\cdot \|$ is the induced matrix norm. 

Let $\breve X^i_T = \max_{1 \leq t \leq T} \| \breve x^i_t \|$ be the maximum
norm of the auxiliary state along the entire trajectory. The next bound follows
from \cite[Lemma 4]{altproof}.
\begin{lemma}\label{lem:Xq}
  For each node $i \in N$, any $q \ge 1$ and any $T >1$, 
  \begin{align*}
   \EXP \biggl[ \frac{(\breve{X}_T^i)^q}{(\breve{\sigma}^i)^q} \biggr] &\leq \mathcal{O} \big( \log T\big)
  \end{align*}
\end{lemma}
The following lemma gives an upper bound on the number of episodes $\breve K_T$. 

\begin{lemma}\label{lem:Xqlog}
For any $q\geq 1$, we have
\begin{align}
  \EXP \biggl[ \frac{(\breve{X}_T^i)^q}{(\breve \sigma^i)^q} \log\biggl(\sum_{t=1}^T (\breve{X}_T^{j_t})^2 \biggr) \biggr] 
  &\le \EXP \biggl[ \frac{(\breve{X}_T^i)^q}{(\breve \sigma^i)^q} 
  \log\biggl(\sum_{t=1}^T \sum_{i \in N} (\breve{X}_T^{i})^2 \biggr) \biggr] 
  \notag \\
  &\leq \tilde{\mathcal{O}}(1)
\end{align}
\end{lemma}
This can be proved along the same lines as \cite[Lemma~5]{altproof}.

\begin{lemma}\label{lem:episodes}
The number of episodes $\breve{K}_T$ is bounded as follows:
\begin{align*}
\breve{K}_T &\leq  \mathcal{O} \left( \sqrt{ (d_x+d_u) T\log \left(   
\sum_{t=1}^{T-1} \frac{(\breve{X}_T^{j_t})^2 }{(\breve{\sigma}^{j_t})^2}\right)} \right)
\end{align*}
\end{lemma}

\begin{proof}
We can follow the same argument as in the proof of Lemma 5 in \cite{altproof}. Let $\breve{\eta} -1$ be the number of times the second stopping criterion is triggered for $\breve{p}_t$. Using the analysis in the proof of Lemma 5 in \cite{altproof}, we can get the following  inequalities:
\begin{align}
  \breve{K}_T &\leq \sqrt{2 \breve{\eta} T},
  \label{eq:macro_ep_bound}
  \\
\det(\breve{\Sigma}^{-1}_{T}) &\geq 2^{\breve{\eta}-1}\det(\breve{\Sigma}^{-1}_1) \ge
2^{\breve{\eta}-1} \breve \lambda_{\min}^d,
  \label{eq:macro_ep_bound2}
\end{align}
where $d = d_x + d_u$ and $\breve \lambda_{\min}$ is the minimum eigenvalue of $\breve \Sigma_1^{-1}$.

Combining~\eqref{eq:macro_ep_bound2} with $\TR( \breve{\Sigma}_T^{-1} )/d \ge
\det(\breve \Sigma_T^{-1})^{1/d}$, we get 
\(
  \TR(\breve \Sigma_T^{-1}) \ge d \breve \lambda_{\min} 2^{(\breve \eta -1)/d}.
\)
Thus,
\begin{equation}
  \label{eq:macro_ep_bound3}
  \eta \le 1 + \frac{d}{\log 2} 
  \log \biggl( \frac{ \TR(\breve \Sigma_T^{-1}) }{ d \breve \lambda_{\min} } \biggr).
\end{equation}

Now, we bound $\TR(\breve \Sigma_T^{-1})$. From \eqref{eq:sigma_breve_update},
we have
\begin{equation}
  \label{eq:macro_ep_bound4}
  \TR (\breve{\Sigma}_{T}^{-1} ) =
  \TR(\breve{\Sigma}_1^{-1}) +  \sum_{t=1}^{T-1} \frac{1}{(\breve{\sigma}^{j_t})^2} 
\underbrace{\TR(\breve{z}_t^{j_t}  (\breve{z}_t^{j_t})^\TRANS}_{= \| \breve z_t^{j_t} \|^2} ). 
\end{equation}
Note that $\|\breve{z}_t^{j_t}\| = \|[I, \breve{G}(\breve{\theta})^\TRANS]^\TRANS \breve{x}_t^{j_t}\| \leq \breve{M}_G \|\breve{x}_t^{j_t}\| \leq \breve{M}_G \breve{X}_T^{j_t}$. Using $\| \breve z_t^{j_t} \|^2 \le \breve M_G^2 (\breve X_T^{j_t})^2$ in~\eqref{eq:macro_ep_bound4} and substituting the resulting bound on $\TR(\breve \Sigma_T^{-1})$ in~\eqref{eq:macro_ep_bound3} and then combining it with the bound on~$\eta$ in~\eqref{eq:macro_ep_bound}, gives the result of the lemma.
\end{proof}

\begin{lemma}\label{lem:episodes_exp}
The expected value of  $\breve{K}_T$ is bounded as follows:
\begin{align*}
\EXP[\breve{K}_T] &\leq  \tilde{\mathcal{O}}\left( \sqrt{(d_x+d_u) T }  \right)
\end{align*}
\end{lemma}

\begin{proof}
  From Lemma~\ref{lem:episodes}, we get
  \begin{align*}
    \EXP[ \breve K_T ] &\le 
    \mathcal{O}\Biggl( \EXP\Biggl[
\sqrt{ (d_x+d_u) T\log \biggl(   
\sum_{t=1}^{T-1} \frac{(\breve{X}_T^{j_t})^2 }{(\breve{\sigma}^{j_t})^2}\biggr)} 
    \Biggr] \Biggr)
    \\
    &\stackrel{(a)}\le
  \mathcal{O}\biggl(\sqrt{ (d_x+d_u) T\log \biggl(
    \EXP \biggl[
\sum_{t=1}^{T-1} \frac{(\breve{X}_T^{j_t})^2 }{(\breve{\sigma}^{j_t})^2} \biggr]\biggr)} \biggr) 
  \\
  &\le
  \mathcal{O}\biggl(\sqrt{ (d_x+d_u) T\log \biggl(
    \EXP \biggl[
\sum_{t=1}^{T-1} \sum_{i \in N} \frac{(\breve{X}_T^{i})^2 }{(\breve{\sigma}^{i})^2} \biggr]\biggr)} \biggr) 
  \\
  &\stackrel{(b)}{\le} \tildeO( \sqrt{(d_x + d_u) T })
  \end{align*}
where $(a)$ follows from  Jensen's inequality and $(b)$ follows from
Lemma~\ref{lem:Xq}. 
\end{proof}

\subsection{Proof of Lemma \ref{lem:breve_regret_terms}}

\begin{proof}

We will prove each part separately. 

1) Bounding $\breve{R}^i_0(T)$: From an argument similar to the proof of Lemma~5 of~\cite{ouyang2019posterior}, we get that 
  \(
    \breve R^i_0(T) \le (\breve \sigma^i)^2 \breve M_J \EXP[ \breve K_T ].
  \)
The result then follows from substituting the bound on $\EXP[ \breve K_T ]$ from Lemma~\ref{lem:episodes_exp}.

2) Bounding $\breve{R}^i_1(T)$:  
\begin{align}
  \breve{R}^i_1(T)
  &= \EXP\biggl[\sum_{k=1}^{\breve K_T}\sum_{t=\breve t_k}^{\breve t_{k+1}-1}  \Big[ (\breve x_t^i)^\TRANS \breve S_k \breve x_t^i - (\breve x_{t+1}^i)^\TRANS \breve S_k \breve x_{t+1}^i\Big]\biggr] \notag \\
  &=\EXP\biggl[\sum_{k=1}^{\breve K_T}\Big[ (\breve x^i_{\breve t_k})^\TRANS \breve S_k \breve x_{\breve t_k}^i -(\breve x_{\breve t_{k+1}}^i )^\TRANS \breve S_k \breve x_{\breve t_{k+1}}^i \Big]\biggr] \notag  \\
  &\leq \EXP\biggl[\sum_{k=1}^{\breve K_T}
  (\breve x^i_{\breve t_k})^\TRANS \breve S_k \breve x_{\breve t_k}^i \biggr]
  \leq \EXP\biggl[\sum_{k=1}^{\breve K_T}
  \| \breve S_k \| \|\breve x_{t_k}^i \|^2 \biggr]
  \notag \\ &
  \leq \breve M_S \EXP[ \breve K_T (\breve X_T^i)^2 ]
  \label{eq:R1bound}
\end{align}
where the last inequality follows from $\| \breve S_k \| \le \breve M_S$. 
Using the bound for $\breve{K}_T$ in Lemma~\ref{lem:episodes}, we get
 \begin{equation}
 \breve{R}^i_1(T) \leq 
 \mathcal{O}\Biggl( \sqrt{ (d_x+d_u) T }\,
   \EXP\Biggl[(\breve X^i_T)^2 
  \sqrt{ \log \biggl(   
\sum_{t=1}^{T-1} \frac{(\breve{X}_T^{j_t})^2 }{(\breve{\sigma}^{j_t})^2}\biggr)} 
    \Biggr] \Biggr).
 \end{equation}
 Now, consider the term
 \begin{align}
   \hskip 1em & \hskip -1em 
   \EXP\Biggl[(\breve X^i_T)^2 
  \sqrt{ \log \biggl(   
\sum_{t=1}^{T-1} \frac{(\breve{X}_T^{j_t})^2 }{(\breve{\sigma}^{j_t})^2}\biggr)} 
    \Biggr] \Biggr)
    \notag \\
    &\stackrel{(a)}\le \sqrt{ \EXP[ (\breve X^i_T)^4 ] \;
    \EXP\biggl[ \log \biggl( 
\sum_{t=1}^{T-1} \frac{(\breve{X}_T^{j_t})^2 }{(\breve{\sigma}^{j_t})^2}\biggr)\biggr]} 
    \notag \\
    &\stackrel{(b)}\le \sqrt{ \EXP[ (\breve X^i_T)^4 ] \;
    \log \biggl( \EXP\biggl[ 
  \sum_{t=1}^{T-1} \sum_{i \in N} \frac{(\breve{X}_T^{i})^2 }{(\breve{\sigma}^{j_t})^2}\biggr]\biggr)} 
  \notag \\
  &\stackrel{(c)}\le \tildeO( (\breve \sigma^i)^2 ),
  \label{eq:R1bound2}
  \end{align}
   where $(a)$ follows from Cauchy-Schwarz,  $(b)$ follows from Jensen's inequality and $(c)$ follows from Lemma~\ref{lem:Xq}. The result then follows from substituting~\eqref{eq:R1bound2} in~\eqref{eq:R1bound}.

3) Bounding $\breve{R}^i_2 (T)$: As in~\cite{ouyang2019posterior}, we can bound the inner summand in $\breve{R}^i_2(T)$
as
\begin{multline*}
  \bigl[(\breve \theta^\TRANS \breve z^i_t)^\TRANS \brMAT S_k (\breve\theta^\TRANS \breve z^i_t) -( \breve \theta_k^\TRANS \breve z^i_t)^\TRANS \brMAT S_k ((\breve
 \theta_k)^\TRANS \breve z^i_t) \bigr] \\
 \le
  \mathcal{O}( \breve X_T^i \|(\breve\theta- \breve\theta_k )^\TRANS \breve z_t^i \|).
\end{multline*}

Therefore, 
\begin{equation}
  \breve R^i_2(T) \leq \mathcal{O}\biggl( \EXP \Big[ \breve X_T^i \sum_{k=1}^{\breve K_T}\sum_{t=\breve t_k}^{\breve t_{k+1}-1}  \| (\breve\theta- \breve\theta_k )^\TRANS \breve z_t^i \| \Big]\biggr).
\label{eq:boundR2_1_breve}
\end{equation}
The term inside $\mathcal{O}(\cdot)$  can be written as
\begin{align}
&\EXP\bigg[\breve X_T^i \sum_{k=1}^{\breve K_T}\sum_{t=\breve t_k}^{\breve t_{k+1}-1}  \|(\breve\theta- \breve\theta_k )^\TRANS \breve z_t^i \| \bigg] \notag\\
&\quad = \EXP\bigg[\breve X_T^i \sum_{k=1}^{\breve K_T}\sum_{t=\breve t_k}^{\breve t_{k+1}-1}  \| (\breve\Sigma^{-0.5}_{\red{t_k}}(\breve\theta-\breve\theta_k ))^\TRANS  \breve\Sigma^{0.5}_{\red{t_k}} \breve z_t^i \| \bigg]
\displaybreak[1]
\notag\\
&\quad\leq \EXP\bigg[\sum_{k=1}^{\breve K_T}\sum_{t=\breve t_k}^{\breve t_{k+1}-1}  \| \breve\Sigma^{-0.5}_{\red{t_k}} (\breve\theta- \breve\theta_k )\| \times  \breve X_T^i \|\breve\Sigma^{0.5}_{\red{t_k}} \breve z_t^i \| \bigg]
\displaybreak[1]
\notag\\
& \quad \leq  \sqrt{\EXP\bigg[\sum_{k=1}^{\breve K_T}\sum_{t=\breve t_k}^{\breve t_{k+1}-1}  \| \breve\Sigma^{-0.5}_{\red{t_k}} (\breve\theta- \breve\theta_k ) \|^2 \bigg]}  \notag\\
&\qquad \times \sqrt{\EXP\bigg[\sum_{k=1}^{\breve K_T}\sum_{t=\breve t_k}^{\breve t_{k+1}-1}  (\breve X_T^i)^2 \| \breve\Sigma^{0.5}_{\red{t_k}} \breve z_t^i \|^2\bigg]}
\label{eq:boundR2_2_breve}
\end{align}
where the last inequality follows from Cauchy-Schwarz inequality. 

Following the same argument as \cite[Lemma 7]{altproof}, the first part of \eqref{eq:boundR2_2_breve} is bounded by
\begin{equation}
  \EXP\bigg[\sum_{k=1}^{\breve K_T}\sum_{t=\breve t_k}^{\breve t_{k+1}-1}  \| \breve\Sigma^{-0.5}_{\red{t_k}} (\breve\theta-\breve\theta_k ) \|^2\bigg]
  \leq \mathcal{O}(d_x(d_x+d_u)T).
\label{eq:boundR2_part1_breve}
\end{equation}

For the second part of the bound in \eqref{eq:boundR2_2_breve}, we follow the same argument as \cite[Lemma 8]{altproof}. 
Recall that $\breve \lambda_{\min}$ is the smallest eigenvalue of $\breve \Sigma_1^{-1}$. Therefore, by~\eqref{eq:sigma_breve_update}, all eigenvalues of $\breve \Sigma_t^{-1}$ are no smaller than $\breve \lambda_{\min}$. Or, equivalently, all eigenvalues of $\breve \Sigma_t$ are no larger than $1/\breve \lambda_{\min}$. 

\red{Using~\cite[Lemma~11]{abbasi2011regret}, we can show that for any $t \in \{t_k, \dots, t_{k+1} - 1\}$, 
\begin{align}
\|\breve{\Sigma}^{0.5}_{t_k} \breve{z}_t^i \|^2
 &=
 (\breve z^i_t)^\TRANS \breve \Sigma_{t_k} \breve z^i_t 
 \le 
 \frac{ \det \breve \Sigma_{t}^{-1} }{\det\breve \Sigma_{t_k}^{-1}}
 (\breve z^i_t)^\TRANS \breve \Sigma_{t} \breve z^i_t 
 \notag \\
 &\le F_1(\breve X^i_T)\,
 (\breve z^i_t)^\TRANS \breve \Sigma_{t} \breve z^i_t 
 \label{eq:new-1}
\end{align}
where $F_1(\breve X^i_T) = 
 \bigl( 1 + ( \breve M_G^2 
 (\breve X^i_T)^2/{\breve \lambda_{\min} \breve \sigma_w^2 }) \bigr)^{\breve T_{\min} \vee 1}$ and 
the last inequality follows from~\cite[Lemma 10]{altproof}.}

Moreover, since all eigenvalues of $\breve \Sigma_t$ are no larger than $1/\breve \lambda_{\min}$, we have
\(
   (\breve z^i_t)^\TRANS \breve \Sigma_t \breve z^i_t 
   \le  \| \breve z^i_t \|^2 /{\breve \lambda_{\min}} 
   \le  \breve M_G^2 (\breve X^i_T)^2/{\breve \lambda_{\min}}
 \). Therefore,
\begin{align}
  \hskip 2em & \hskip -2em 
  (\breve z^i_t)^\TRANS \breve \Sigma_t \breve z^i_t 
  \le
  \biggl((\breve \sigma^i)^2 \vee \frac{\breve M_G^2 (\breve X^i_T)^2 }{ \breve \lambda_{\min}} \biggr)
  \biggl(1 \wedge \frac{(\breve z^i_t)^\TRANS \breve \Sigma_t \breve z^i_t}{(\breve \sigma^i)^2}\biggr)
  \notag\\
  &\le \biggl((\breve \sigma^i)^2 + \frac{\breve M_G^2 (\breve X^i_T)^2 }{ \breve \lambda_{\min}} \biggr)
  \biggl(1 \wedge
  \frac{(\breve z^{j_t}_t)^\TRANS \breve \Sigma_t \breve z^{j_t}_t}{(\breve \sigma^{j_t})^2}
  \biggr),
  \label{eq:new-2}
\end{align}
where the last inequality follows from the definition of~$j_t$. \red{Let $F_2(\breve X^i_T) = 
\bigl((\breve \sigma^i)^2 + ({ \breve \lambda_{\min}}/{\breve M_G^2 (\breve X^i_T)^2 }) \bigr)$. Then,
\begin{align}
  \hskip 1em & \hskip -1em 
  \sum_{t=1}^T (\breve z^i_t)^\TRANS \breve \Sigma_t \breve z^i_t 
  \le F_2(\breve X^i_T) \sum_{t=1}^T
  \biggl(1 \wedge
  \frac{(\breve z^{j_t}_t)^\TRANS \breve \Sigma_t \breve z^{j_t}_t}{(\breve \sigma^{j_t})^2}
  \biggr)
  \notag \\
  &= F_2(\breve X^i_T)
\sum_{t=1}^T \biggl(1 \wedge \biggl\lVert\frac{\Sigma^{0.5}_t\breve z^{j_t}_t(\breve z^{j_t}_t)^\TRANS \Sigma^{0.5}_t}{(\breve \sigma^{j_t})^2}\biggr\rVert \biggr) \notag \\
  &\stackrel{(a)}\le F_2(\breve X^i_T)
  \biggl[ 2d \log \bigg( \frac{ \TR( \breve \Sigma_{T+1}^{-1 }) }{ d } \biggr) - \log\det \Sigma^{-1}_1\biggr]
  \notag\\
  &\stackrel{(b)}\le F_2(\breve X^i_T) \biggl[
   2 d \log\biggl(\frac 1{d} \biggl(
    \TR(\breve \Sigma_1^{-1}) + \breve M_G \sum_{t=1}^T \frac{(\breve X^{j_t}_T)^2}{ (\breve \sigma^{j_t})^2} \biggr)
    \biggr) 
    \notag \\
  & \hskip 6em
  - \log \det \Sigma^{-1}_1 \biggr]
    \label{eq:new-3}
\end{align}
where $d = d_x + d_u$ and $(a)$ follows from~\eqref{eq:sigma_breve_update} and
the intermediate step in the proof of~\cite[Lemma 6]{abbasiyadkori2014bayesian}.
and $(b)$ follows from~\eqref{eq:macro_ep_bound4} and the subsequent discussion.}

\red{Using~\eqref{eq:new-1} and~\eqref{eq:new-3}, we can bound the second term of~\eqref{eq:boundR2_2_breve} as follows
\begin{align}
  &
  \EXP\Big[\sum_{t=1}^T (\breve{X}_T^i)^2\|\breve{\Sigma}^{0.5}_{t_k} \breve{z}_t^i \|^2\Big] 
  \le \mathcal{O} \biggl(d \, \EXP \biggl[ 
      F_1(\breve X^i_t) F_2(\breve X^i_T) (\breve X^i_T)^2
  \notag \\
  & \hskip 12em \times 
\log \biggl(   
\sum_{t=1}^{T} (\breve{X}_T^{j_t})^2 \biggr) 
 \biggr] \biggr)
  \notag \\
  &\le \mathcal{O}\biggl(d (\breve \sigma^i)^4 \EXP\biggl[ 
F_1(\breve X^i_T) \frac{F_2(\breve X^i_T)}{(\breve \sigma^i)^2}
 \frac{(\breve X^i_T)^2}{(\breve \sigma^i)^2}
\log \biggl(   
\sum_{t=1}^{T} (\breve{X}_T^{j_t})^2 \biggr) 
 \biggr]\biggr)
  \notag \\
  &\le \tildeO(d (\breve \sigma^i)^4)
  \label{eq:new-4}
\end{align}
where the last inequality follows by observing that 
$F_1(\breve X^i_T) \frac{F_2(\breve X^i_T)}{(\breve \sigma^i)^2}
 \frac{(\breve X^i_T)^2}{(\breve \sigma^i)^2}
\log \bigl(   
\sum_{t=1}^{T} (\breve{X}_T^{j_t})^2 \bigr)$ 
 is a polynomial in $\breve X^i_T/\breve \sigma^i$ multiplied by $\log(\sum_{t=1}^T (X^{j_i}_T)^2)$ and, using Lemma~\ref{lem:Xqlog}.}

The result then follows by substituting \eqref{eq:boundR2_part1_breve} and \eqref{eq:new-4} in~\eqref{eq:boundR2_2_breve}.

\end{proof}

\vfill
\begin{IEEEbiography}{Sagar Sudhakara} received M.Tech Degree in the area of Communication and Signal Processing from Indian Institute of Technology Bombay, Mumbai, India, in 2016 and is currently pursuing PhD in Electrical and Computer Engineering at University of Southern California. He is a recipient of USC Annenberg fellowship and his research interests include reinforcement learning and decentralized stochastic control.
\end{IEEEbiography}

\begin{IEEEbiography}{Aditya Mahajan}
(S’06-M’09-SM’14) is Associate Professor in the the department of Electrical
and Computer Engineering, McGill University, Montreal, Canada. He currently serves as
Associate Editor of Mathematics of Control, Signal, and Systems.
He is the recipient of the 2015 George
Axelby Outstanding Paper Award, 2014 CDC Best Student Paper Award (as
supervisor), and the 2016 NecSys Best Student Paper Award (as supervisor). 
His principal research interests include learning and control of centralized and decentralized stochastic systems.
\end{IEEEbiography}

\begin{IEEEbiography}{Ashutosh Nayyar}(S’09-M’11-SM’18) received the
B.Tech. degree in electrical engineering from the
Indian Institute of Technology, Delhi, India, in 2006.
He received the M.S. degree in electrical engineering
and computer science in 2008, the MS degree in
applied mathematics in 2011, and the Ph.D. degree
in electrical engineering and computer science in
2011, all from the University of Michigan, Ann
Arbor. He was a Post-Doctoral Researcher at the
University of Illinois at Urbana-Champaign and at
the University of California, Berkeley before joining
the University of Southern California in 2014. His research interests are
in decentralized stochastic control, decentralized decision-making in sensing
and communication systems, reinforcement learning, game theory, mechanism
design and electric energy systems.
\end{IEEEbiography}

\begin{IEEEbiography}{Yi Ouyang} received the B.S. degree
in Electrical Engineering from the National Taiwan
University, Taipei, Taiwan in 2009, and the M.Sc
and Ph.D. in Electrical Engineering and Computer Science at the University
of Michigan, in 2012 and 2015, respectively. He is
currently a researcher at Preferred Networks, Burlingame, CA. His research interests include reinforcement
learning, stochastic control, and stochastic dynamic
games.
\end{IEEEbiography}

\end{document}